\documentclass[12pt]{article}
\usepackage{amsmath}
\usepackage{graphicx,psfrag,epsf}
\usepackage{enumerate}
\usepackage[sort]{natbib}
\usepackage{url} 
\usepackage{bm}
\usepackage{graphicx}

\usepackage{float}
\usepackage{subcaption}
\usepackage{amsthm}

\usepackage{color}
\usepackage{amssymb}
\usepackage{hyperref}
\usepackage{caption}
\usepackage[T1]{fontenc}
\usepackage{babel}
\usepackage[section]{placeins}
\usepackage{graphicx,multicol, caption, xcolor}
\usepackage{lipsum} 

\newcommand{\blind}{1}

\newtheorem{theorem}{Theorem}
\newtheorem{lemma}{Lemma}
\newtheorem{corollary}{Corollary}
\newtheorem{proposition}{Proposition}
\newcommand{\red}[1]{\textcolor{black}{#1}}

\newcommand{\cc}{{C}}
\newcommand{\cct}{\tilde{C}}
\newcommand{\loss}{\mathcal{L}} 
\newcommand{\R}{\mathcal{R}} 
\newcommand{\bbeta}{\bm{\beta}}
\newcommand{\prde}{\textrm{prde}}
\newcommand{\prdes}{\textrm{prdes}}

\newcommand{\phatu}{\hat p^\texttt{U}}

\newcommand{\ex}{\mathbb{E}}
\newcommand{\ap}{{\bm{x}}}
\newcommand{\af}{\tilde{\bm{x}}}
\newcommand{\Kf}{\tilde{K}}
\newcommand{\phat}{\hat{p}}
\usepackage{comment}

\begin{document}

\def\spacingset#1{\renewcommand{\baselinestretch}%
{#1}\small\normalsize} \spacingset{1.45}


\if1\blind
{
  \title{\bf Empirical Bayes Predictive Density Estimation under Covariate Shift in Large Imbalanced Linear Mixed Models}
  \author{Abir Sarkar \\ 
  Department of Statistics and Data Science,\\ Cornell University \\
  and,\\
  Gourab Mukherjee\\
    Department of Data Sciences and Operations,\\ University of Southern California \\
      and,\\
Keisuke Yano\\
Department of Fundamental Statistical Mathematics,\\
The Institute of Statistical Mathematics
    }
  \maketitle
} \fi

\begin{abstract}
We study empirical Bayes (EB) predictive density estimation in linear mixed models (LMMs) with large number of units, which induce a high-dimensional random‐effects space. Focusing on Kullback–Leibler (KL) risk minimization, we develop a calibration framework to optimally tune predictive densities derived from on a broad class of flexible priors. Our proposed method addresses two key challenges in predictive inference: (a) severe data scarcity leading to highly imbalanced designs, in which replicates are available for only a small subset of units; and (b) distributional shifts in future covariates.

To estimate predictive KL risk in LMMs, we use a data-fission approach that leverages exchangeability in the covariate distribution. We establish convergence rates for our proposed risk estimators and show how their efficiency deteriorates as data scarcity increases. Our results imply the decision-theoretic optimality of the proposed EB predictive density estimator. The theoretical development relies on a novel probabilistic analysis of the interaction between data fission, sample reuse, and the predictive heat-equation representation of \citet{george2006improved}, which expresses predictive KL risk through expected log-marginals. Extensive simulation studies demonstrate strong predictive performance and robustness of the proposed approach across diverse regimes with varying degrees of data scarcity and covariate shift.
\end{abstract}

\noindent%
\vfill

\newpage
\spacingset{1.45} 

\section{Introduction}\label{sec:intro}
We consider predictive density estimation \citep{Aitchison-book, Geisser-book} in linear mixed-effects models (LMMs) \citep{pinheiro2000linear,demidenko2013mixed}. Our focus is on LMMs with a large number of units that need the use of large number of random effects. Such models are increasingly prevalent in contemporary applications (See \citet{gao2020estimation}, \citet{ghosh2022backfitting} and  the references therein). We also allow for the possibility that past and future data may not be exchangeable \citep{cauchois2024robust}, with future data involving covariates drawn from a different distribution (See Section 2 of \citet{tibshirani2019conformal}). This phenomenon, known as covariate shift, frequently arises in modern predictive inference problems \citep{gibbs2021adaptive}. Standard predictive methods typically perform poorly under these regimes \citep{yang2024doubly,angelopoulos2022private,gibbs2024online}.

Our objective is to estimate the future conditional density of the observations, commonly referred to as the predictive density \citep{Geisser-book,Aitchison-book}. Existing approaches to predictive inference under covariate shift typically proceed by estimating the density ratio between historical and future covariate distributions and subsequently reweighting or adjusting the predictors. In contrast, we directly target minimization of the predictive Kullback–Leibler (KL) loss between the true predictive density and its estimator. This loss explicitly integrates over the future covariate distribution; consequently, minimizing it yields predictive density estimators that adapt naturally to the covariate distribution of the target population. The KL loss is a well-studied metric in predictive density estimation and is known to produce estimators with favorable risk properties \citep{Komaki01,George08,Mukherjee-thesis,Ghosh08,maruyama2012bayesian,Fourdrinier11,marchand2018predictive,Rockova2024Adaptive}.

We adopt an empirical Bayes framework and seek predictive densities that minimize the predictive KL risk over a broad class of estimators indexed by flexible prior families in LMMs. To this end, we construct predictive KL risk estimators that are uniformly well behaved over the entire candidate class and use them for risk minimization.

This setting presents several fundamental challenges. Using the predictive heat-equation representation (see Equation (16) of \citet{george2006improved}), the predictive KL loss can be expressed as an expectation of the predictive marginal log-likelihood. Unlike point estimation, this quantity does not admit canonical unbiased risk estimators unless the prior family is Gaussian. In this work, we consider predictive density estimators (\prde) arising from both Gaussian and non-Gaussian priors.

To address this difficulty, we employ the sample splitting approach. For point estimation in sequence models \citet{Brown2013Poisson} prescribed a sample splitting strategy that has been subsequently extended to complex settings (see \citet{Leiner2023DataFission,Dharamshi2024Generalized,Neufeld2024DataThinning} and references therein). The role of sample splitting in predictive density estimation for LMMs has not been previously studied. We develop new theoretical results that enable efficient estimation of the predictive KL risk under covariate shifts via sample splitting. We next describe our predictive setup. Detailed background, references, and a full discussion of our contributions are provided subsequently.


\newcommand{\Y}{\tilde Y}
\newcommand{\y}{\tilde y}
\newcommand{\cb}{\mathbb{C}}

\subsection{Predictive Set-up: LMM with Covariate Shift}
For each unit $i=1,\ldots,n$ in the linear mixed effects model consider observing $K_i$ replicates $\{Y_{ik}: k=1,\ldots, K_i\}$. The expected value of the response $Y_{ik}$ is linearly related to covariates  $\ap_{ik}$ and $u_{ik}$ by  parameters $\bbeta \in \mathbb{R}^d$ and $\gamma_i \in \mathbb{R}$ respectively. The model is:
\begin{align}\label{eq:model.11}
    y_{ik}=\ap_{ik}'\bbeta + u_{ik} \gamma_i +\sigma\, \epsilon_{ik} \text{ for } k=1,\ldots,K_i \text{ and, }i=1,\ldots,n.
\end{align}
The noise terms $\epsilon_{ik}$'s are independently distributed from  standard normal distribution. Based on observing $\{y_{ik}:k=1,\ldots,K_i;\,i=1,\ldots,n\}$, consider predicting future observations $\y_{ik}$ from units $i=1,\ldots,n$ based on the following model:  
\begin{align}\label{eq:model.12}
    \y_{ik}=\af_{ik}'\bbeta +  v_{ik} \gamma_i +\sigma\, \tilde{\epsilon}_{ik} \text{ for } k=1,\ldots,\Kf_i \text{ and, }i=1,\ldots,n.
\end{align}
The covariates in \eqref{eq:model.11} and \eqref{eq:model.12} can differ, but all the parameters $\bbeta, \sigma, \gamma_i$ are invariant. The noise terms $\tilde{\epsilon}_{ik}$s have standard normal distributions and are independent not only across $(i,k)$ pairs in  \eqref{eq:model.12} but also with the noise terms in \eqref{eq:model.11}. The variance $\sigma^2$ is unknown. We further assume that the random effects $\bm{\gamma}=\{\gamma_i:1\leq i\leq n\}$ in \eqref{eq:model.11}-\eqref{eq:model.12} are independent and identically distributed (i.i.d.) from a unknown distribution $g_0$.

Let $\bm{\theta}=(\bbeta,\sigma,\bm{\gamma})$ denote the full parameter vector, and let
$\cc=\{(x_{ik},u_{ik}) : 1\le k\le K_i,\; i=1,\ldots,n\}$
and
$\cct=\{(\tilde x_{ik},v_{ik}) : 1\le k\le \Kf_i,\; i=1,\ldots,n\}$
denote the covariates in \eqref{eq:model.11} and \eqref{eq:model.12}, respectively. Let $\Theta$ be the parameter space. If the unknown parameter $\bm{\theta}\in\Theta$ were known, the true future density is
\[
p(\bm{\Y}\mid\bm{\theta},\cct)
=\prod_{i=1}^{n}\prod_{k=1}^{\Kf_i}
\phi\!\left(\tilde y_{ik};\,\af_{ik}'\bbeta+v_{ik}\gamma_i,\sigma\right).
\]
where $\phi(\,\cdot\,;\mathtt{me},\mathtt{sd})$ denotes the normal density with mean $\mathtt{me}$ and standard deviation $\mathtt{sd}$. Our goal is to construct an estimator $\hat{p}(\bm{\Y}\mid\bm{Y},\cc,\cct)$ of the future conditional density of
$\bm{\Y}=\{\tilde y_{ik}: i=1,\ldots,n;\, k=1,\ldots,\Kf_i\}$
based on the observed data
$\bm{Y}=\{y_{ik}: i=1,\ldots,n;\, k=1,\ldots,K_i\}$.
To evaluate the performance of $\hat{p}$, we use the information divergence of \citet{Kullback51} between the true future density and its estimate as the loss function:
\[
\loss_n[\cc,\cct](\bm{\theta},\hat{p})
=\Big(\sum_{i=1}^n \Kf_i\Big)^{-1}
\int p(\bm{\Y}\mid\bm{\theta},\cct)
\log\!\left(
\frac{p(\bm{\Y}\mid\bm{\theta},\cct)}
{\hat{p}(\bm{\Y}\mid\bm{y},\cc,\cct)}
\right)
\,d\bm{\Y}.
\]
Henceforth, we denote $\cc\cup\cct$ by $\cb$. The corresponding predictive KL risk is obtained by averaging over the distribution of the past observations:
\begin{align}\label{eq:f.risk}
\R_n[\cb](\bm{\theta},\hat{p})
=\int p(\bm{Y}\mid\bm{\theta},\cc)\,
\loss_n[\cb](\bm{\theta},\hat{p})\,
d\bm{Y}.
\end{align}
The predictive KL risk in \eqref{eq:f.risk} depends explicitly on the covariates $\cc$ and $\cct$, and thus reflects changes arising from shifts between the past and future covariate distributions. 

We consider linear mixed models with a large number of units. Our focus is on data-scarce settings with limited replication, where only a small subset of units contributes multiple observations, while the majority of units are observed at most once.
Let $\eta_n$ denote the proportion of units with two or more replicates, defined as
\[
\eta_n
= n^{-1}\,\mathrm{cardinality}\bigl\{\, i : 1\le i\le n \text{ and } K_i>1 \,\bigr\}.
\]
We are particularly interested in regimes where $\eta_n$ is small, and especially in the asymptotic setting where $\eta_n\to 0$ as $n\to\infty$.

\newcommand{\be}{\bm{\mathrm{b}}}
\newcommand{\se}{\mathrm{s}}

\subsection{Empirical Bayes Predictive Density Estimation}
Consider a broad class $\mathcal{G}$ of probability distributions for the random effects $\bm{\gamma}$ defined in equations \eqref{eq:model.11}--\eqref{eq:model.12}. For a fixed parameter set $(\bm{\beta}, \sigma) = (\bm{b}, s)$, and for any prior distribution $g \in \mathcal{G}$, the conditional Bayes \prdes\  is given by:
\begin{align}\label{eq:bayes.prde1}
\hat{p}[\be,\se,g](\bm{\Y};\bm{Y},\cb)
=\int {p}(\bm{\Y}\mid \be,\se,\bm{\gamma},\cct)\,
g[\be,\se,\cc](\bm{\gamma}\mid \bm{Y})\,d\bm{\gamma},
\end{align}
where, $g[\be,\se,\cc](\bm{\gamma}\mid \bm{Y})$ is the posterior distribution of $\bm{\gamma}$ based on the observing $\bm{Y}$ for fixed  $\bbeta=\be$ and $\sigma=\se$. It is given by
\[
\bigl(m[\be,\se,g](\bm{Y};\cc)\bigr)^{-1}
\prod_{i=1}^n g(\gamma_i)
\prod_{k=1}^{K_i}
{p}(\bm{Y}_{ik}\mid \be,\se,\gamma_i,u_{ik},a_{ik}).
\]
Here, $m[\be,\se,g](\bm{Y};\cc)$ is the corresponding normalizing constant, representing the marginal conditional likelihood of $\bm{Y}$ under the prior $g$ on $\bm{\gamma}$. Consider the following class of \prdes:
\[
\mathcal{P}_{\mathcal{G}}
=\left\{
\hat{p}[\be,\se,g]:
g\in\mathcal{G},\;
\se>0,\;
\be\in\mathbb{R}^d
\right\}.
\]

We adopt the empirical Bayes framework of \citet{efron2012large} and consider estimation of the hyperparameter within the class $\mathcal{P}_{\mathcal{G}}$. In particular, we focus on estimating the distribution $g$ that governs the exchangeable hierarchical structure of $\bm{\gamma}$. For a comprehensive review of empirical Bayes methodologies, see \citet{xie2012sure}.

In this work, we propose an empirical Bayes \prde\, by selecting a member of $\mathcal{P}_{\mathcal{G}}$ that minimizes (or approximately minimizes) the predictive KL risk \eqref{eq:f.risk} over this class. As a natural baseline, one may consider using non-informative priors for $\bm{\gamma}$ in \eqref{eq:bayes.prde1}. We show that the proposed EB \prde\, substantially outperforms all such uniform-prior-based \prdes. Moreover, even when the true density $g_0 \notin \mathcal{G}$, the inherent richness of the class $\mathcal{P}_{\mathcal{G}}$, combined with our flexible EB risk-estimation-based approach, often yields a \prde\, whose risk is close to the optimal Bayes risk.

\subsection{Our Contributions}
Our main contributions are as follows.

(a) The roles of shrinkage in \prde\ has been studied primarily in Gaussian sequence models  \citep{Mukherjee-thesis,Mukherjee-15,Liang-thesis,Xu10,marchand2018predictive,Rockova2024Adaptive}. We extend this line of work by studying optimal tuning of shrinkage priors in linear mixed models allowing for non-Gaussian random effects.

(b) Risk-estimation-based empirical Bayes \prde\ has been developed in \citet{Xu12} and \citet{george2021optimal} for invariant sequence models. In contrast, we study \prde\ in a covariate-shifted LMM framework, where the parameters are time-invariant but the design evolves over time \citep{shimodaira2000improving}. The settings in \citet{George08} and \citet{MukherjeeYanoIJPAM} are closest to ours. 
The former can be extended to Gaussian random effects, but not to the more general distributions considered here. The latter assumes known fixed effects and scale parameters for \prde, imposes stronger untestable assumptions on the design, and does not study the risk properties of empirical Bayes \prde\ under varying covariate distributions or in regimes with differing levels of data scarcity.

(c) To estimate the predictive Kullback–Leibler (KL) risk, we build on the predictive heat-equation representation of \citet[Eq.~(16)]{george2006improved}. However, in the presence of non-Gaussian random effects, this representation does not yield a tractable unbiased risk estimator (see \ref{eq:unbiased.gaussian}). We therefore adopt a sample-splitting strategy to construct feasible KL risk estimates. Existing sample-splitting and data-fission techniques \citep{Brown2013Poisson,Leiner2023DataFission,Dharamshi2024Generalized,Neufeld2024DataThinning} are not directly applicable in the predictive KL setting, as these procedures induce variance inflation in the observed data, whereas predictive KL risk is highly sensitive to the ratio of past and future variances \citep{Mukherjee-15,YKK21}. To address this limitation, we develop a novel sample-reuse scheme based on data fission that exploits the exchangeability of the random effects. By carefully controlling the dependence structure introduced through sample reuse, we establish rates of convergence for the resulting risk estimators. This approach contributes new mathematical statistics perspectives to the predictive sample-reuse literature (Chapter~5 of \citet{Geisser-book}).

(d) In the predictive framework of \eqref{eq:model.11}–\eqref{eq:model.12}, $g$-modeling approaches \citep{Efron2014TwoModeling} that estimate the random-effects distribution and plug it  for \prde\ can be substantially suboptimal under covariate shift. In contrast, our predictive KL risk–based procedure adapts to the target covariate distribution and is asymptotically oracle optimal (Theorem~\ref{thm.risk.conv} and Corollary~\ref{thm.main_criteria_result}); numerical evidence is provided in Sections~\ref{sec:simulation}.

\subsection{Organization of the paper}
The paper is organized as follows. Section~\ref{sec:section2} introduces the proposed empirical Bayes predictive density estimator based on risk estimation. Section~\ref{sec:theory} establishes its decision-theoretic properties, and Section~\ref{sec:simulation} describes implementation details. Section~\ref{sec:discussion} presents numerical experiments illustrating predictive performance across a range of regimes. All proofs are deferred to the \hyperref[sec:appendix]{Appendix}.

\section{Proposed EB Predictive Density Estimator}\label{sec:section2}
Consider the \prde\ $p^{\textsf{OR}} \in \mathcal{P}_{\mathcal{G}}$ that minimizes the KL predictive loss in \eqref{eq:f.risk}:
\begin{align}\label{eq:oracle}
p^{\textsf{OR}}
=\arg\min_{\hat p \in \mathcal{P}_{\mathcal{G}}}
\R_n[\cb](\bm{\theta}, \hat{p}).
\end{align}
We refer to $p^{\textsf{OR}}$ as the \emph{oracle} predictor, since it depends on the unknown parameter vector $\bm{\theta}$ and therefore is not implementable in practice. It provides a theoretical lower bound on the achievable predictive loss within the class $\mathcal{P}_{\mathcal{G}}$. 

Our objective is to construct data-driven predictors by minimizing risk estimates that do not depend on unknown parameters. We then compare their risk properties with those of the ideal oracle predictor $p^{\textsf{OR}}$. In this sense, the oracle predictor serves as a benchmark for evaluating the performance of practical predictors.

Under model \eqref{eq:model.11}, constructing asymptotically consistent estimators of $\bm{\beta}$ and $\sigma$ is straightforward. To develop the proposed \prde\ estimator, we consider estimators $\hat{\bm{\beta}}$ and $\hat{\sigma}$ that are $O_p(n^{-\alpha})$-consistent for some $\alpha \in (0,1/2]$.
Details on their construction are provided later in Sec~\ref{sec.3.1}. 
We next consider the risk $\R_n[\cb](\bm{\theta}, g)$ of \prdes\, $\phat[\hat{\bm{\beta}},\hat{\sigma},g]$ which are based on those estimates.  
We estimate these risk functions by  $\hat \R_n[\cb](g)$ for all $g \in \mathcal{G}$.
We will describe the risk estimation elaborately later and show that these risk estimators are uniformly accurate over $\mathcal G$ as $n \to \infty$.
Finally, by minimizing the estimated risk we propose the predictor $\phat[\hat{\bm{\beta}},\hat{\sigma},\hat g]$ from $\mathcal{P}_\mathcal{G}$ as our EB \prde, where,
\begin{align}\label{eq:proposed}
\hat g
= \arg\min_{g \in \mathcal{G}}
\hat \R_n[\cb](g).
\end{align}

To arrive at the risk estimates, we use mathematical identities from \citet{george2006improved} and express the risk of $\phat[\hat{\bbeta},\hat{\sigma},g]$ as a tractable (though still complex) difference of log-likelihoods. Consider the following functional of the observations that also depend on $(\bbeta,\sigma)$: 
$$Z_i(\bbeta,\sigma)=\sigma^{-1} u_{i}^{-2} \bigg(\sum_{k=1}^{K_i} u_{ik}(Y_{ik}-\bm{x}'_{ik}\bbeta)\bigg) \text{ where, } u_i=\bigg(\sum_{k=1}^{K_i} u^2_{ik}\bigg)^{1/2}$$
for $i=1,\ldots,n$. Note that, in \eqref{eq:model.11} when $\bbeta, \sigma$ are known, $Z_i(\bbeta,\sigma)$ is the UMVUE for point estimate of $\gamma_i$ under squared error loss. 
Next, for  $i=1,\ldots,n$ and $k=1,\ldots,\Kf_i$, consider the functional $\tilde Z_{ik}(\bbeta,\sigma)=\sigma^{-1}(\Y_{ik}-\af'_{ik}\bbeta)$ of the predictants in \eqref{eq:model.12} and the following linear combination of these two functionals: 
$$\tilde W_{ik}(\bbeta,\sigma)=v_{ik}(u_i^{2}+v_{ik}^2)^{-1} \{Z_i(\bbeta,\sigma) + v_{ik} \tilde Z_{ik}(\bbeta,\sigma)\}.$$
The genesis for $\tilde W_{ik}(\bbeta,\sigma)$ is more complicated and can be traced to \citet{Komaki01,george2006improved}. It is the UMVUE for point estimation of $v_{ik}\times \gamma_i$ where $\{Y_{ik}: k=1,\ldots,K_i\}$ as well as $\{\tilde Y_{ik}: k=1,\ldots,\Kf_i\}$ are observed. We will use $\tilde W_{ik}$ to express the predictive log-likelihood, it is not known as $\tilde{\bm{Y}}$ is not observed. We use tilde to express this aspect. 

Next, for any $g \in \mathcal{G}$,  consider the following two Gaussian convolution densities:
\begin{align}
m_g(a;u,v,\sigma)&=\int \phi(a; v \gamma \sigma^{-1}, u^{-1}v ) g(\gamma) d\gamma \text{ and, } \label{m.defined.1}\\
\tilde m_g(a;u,v,\sigma)&=
\int \phi(a; v\gamma \sigma^{-1},{
(1+u^{2}v^{-2})^{-1/2}) g(\gamma) d\gamma}~. \label{m.defined.2}
\end{align}
The following result characterizes the predictive risk of 
$\hat p[\hat{\bm{\beta}}, \hat{\sigma}, g]$ in terms of the preceding quantities. 
We first formalize the required assumptions. 

\noindent\textbf{A1.}\label{assumption1}
For some \(\alpha \in (0, 1/2]\), the fixed effects estimator satisfies \(n^{\alpha}\|\hat{\bm{\beta}} - \bm{\beta}\| \to 0\) in \(L_4\). The variance estimator satisfies \(\hat{\sigma}^{-1} - \sigma^{-1} = O_p(n^{-\alpha})\) and \(\hat{\sigma} \geq c\) a.e. for some \(c > 0\).

\noindent\textbf{A2.}\label{assumption2}
The covariates $x, \tilde{x}, u, v$ are uniformly bounded away from zero and from infinity, and
$\limsup_{n \to \infty} 
n^{-1} \sum_{i=1}^n \gamma_i^4 < \infty$.

\noindent\textbf{A3.}\label{assumption3}
The prediction objective spans across all units and is not dominated by a subset. In particular, 
$1 \le \inf_i \Kf_i \le \sup_i \Kf_i < \infty$.

Assumption~\hyperref[assumption1]{A1} ensures that the fixed effects and scale estimators are sufficiently accurate so that the corresponding plug-in errors are asymptotically negligible. As we are dealing with the risk (expected loss) function, we require \(L_r\) convergence of the fixed effects estimator. In the following subsection, we show that constructing estimators satisfying \hyperref[assumption1]{A1} is straightforward.
Assumption~\hyperref[assumption2]{A2} rules out heavy-tailed or unbounded designs, enabling uniform laws of large numbers and central limit theorems without domination by outliers. Note that, \hyperref[assumption2]{A2} is imposed primarily to streamline the technical arguments; it can be weakened at the expense of additional technical complexity.   \hyperref[assumption3]{A3} prevents the prediction target from being degenerate or driven by a few units, ensuring stable cross-sectional aggregation.

\begin{theorem}\label{prop.bayes.risk}
For any fixed $\bbeta$ and $\sigma>0$, under assumptions~\hyperref[assumption1]{A1}-~\hyperref[assumption3]{A3} for some $\alpha \in [0,1/2]$ the risk of $\hat p[\hat{\bm{\beta}},\hat \sigma, g]$ is given by:
\begin{align}
&\,\R_n[\bm{\theta},\cb](g) = a_n[\cb] +  R_{1,n}[\bm \theta,\cb](g)-R_{2,n}[\bm \theta,\cb](g) + O_p(n^{-\alpha}) \text{ as } n \to \infty, \label{eq:risk.2}\\  
&\text{ where, } a_n[\cb]=\frac 1 {2\,\kappa_n}\sum_{i=1}^n \sum_{k=1}^{\Kf_i}\log\bigg(1+\frac{v_{ik}^2}{u_i^2}\bigg)
\text{ and, } \kappa_n=\sum_{i=1}^n \Kf_i~, \label{eq:risk.2.0}\\ 
&R_{1,n}[\bm \theta,\cb](g)=\frac 1 {\kappa_n}\sum_{i=1}^n \sum_{k=1}^{\Kf_i} \ex_{\gamma_i}\big\{\log m_g(v_{ik}Z_i(\bm \beta,\sigma);u_i,v_{ik},\sigma)\big\}~,\label{eq:risk.2.1}\\ 
&R_{2,n}[\bm \theta,\cb](g)= \frac 1 {\kappa_n}\sum_{i=1}^n \sum_{k=1}^{\Kf_i} \ex_{\gamma_i} \big\{\log \tilde m_g(\tilde W_{ik}(\bm \beta,\sigma);u_i,v_{ik}, \sigma)\big\}~, \label{eq:risk.2.2}
\end{align}
where, the expectations are based on the model \eqref{eq:model.11}-\eqref{eq:model.12} and the functions $m_g$ and $\tilde m_g$ are defined in \eqref{m.defined.1}-\eqref{m.defined.2}. 
\end{theorem}

The first term $a_n$ in the right side of \eqref{eq:risk.2} does not involve any parameters and solely depend on the design. The other two terms depends on the parameters. In \eqref{eq:risk.2.1}-\eqref{eq:risk.2.2}, the expectations within the sum splits out separately for each $\{\gamma_i: i=1,\ldots,n\}$. They only depend on $\gamma_i$ and $\sigma$ and not of $\bm{\beta}$ as the variables are centered. 

The proof is complicated and is provided in the appendix. It built on results in \citet{george2006improved}. Note that, if $(\bm{\beta}, \sigma)$ were known, then for the risk of $\hat p[\bm \beta,\sigma,g]$ we have an exact equality (non-asymptotic) in \eqref{eq:risk.2}. Also, for this case we do not need the assumptions. 
This is an intermediate step in our proof, for which Proposition 4 in \cite{MukherjeeYanoIJPAM} can be applied. However, obtaining the sharper result presented here under unknown parameters requires substantial asymptotic analysis.

When $g$ is the uniform distribution, then the second and third term in \eqref{eq:risk.2} vanishes. Thus, the asymptotic risk of $\phatu:=\hat p[\hat{\bm{\beta}},\hat \sigma, \textsf{U}]\sim a_n$ which is invariant across parameters and  serve as a benchmark. We would like to minimize the excess risk of Bayes \prdes\ $\hat p[\hat{\bm{\beta}},\hat \sigma,{g}]$ with respect to $\phatu$  over the class $\mathcal{G}$. However, as $\bm{\gamma},\sigma$ are unknown, direct evaluations of $R_{1,n}[\bm \theta,\cb](g)$ and $R_{2,n}[\bm \theta,\cb](g)$ are not possible. We propose to estimate them  by $\hat R_{1,n}^{\,g}$ and $\hat R_{2,n}^{\,g}$ respectively. 

We consider the following estimator for $R_{1,n}[\bm \theta,\cb](g)$: 
\begin{align}\label{eq:r1.hat}
\hat R_{1,n}^{\,g}=\bigg(\sum_{i=1}^n \Kf_i\bigg)^{-1}\sum_{i=1}^n \sum_{k=1}^{\Kf_i} \log m_g(v_{ik}Z_i(\hat{\bm{\beta}},\hat \sigma);u_i,v_{ik},\hat\sigma).
\end{align}
For well-behaved prior classes $\mathcal{G}$, for any fixed $g$ we show that $\hat R_{1,n}^{\,g}$ convergences to the true $R_{1,n}[\bm \theta,\cb](g)$ at a near-parametric rate,  
$|R_{1,n}[\bm \theta,\cb](g)- \hat R_{1,n}^{\,g}|=O_{p}(\max\{n^{-1/2}, n^{-\alpha}\})$ as $n \to \infty$. A formal proof is presented within the proof of Theorem~\ref{thm.risk.conv} in Sec~\ref{sec:theory}. 

To provide a reasonable estimate of $R_{2,n}[\bm \theta,\cb](g)$ is straightforward. We can not use the approach used in constructing $\hat R_{1,n}$ as unlike $R_{1,n}$, $R_{2,n}$ involves the predictive log-likelihood through $\tilde W_{ik}(\bm \beta,\sigma)$ \citep{George12}. Note that, even if $\bm \beta$, $\sigma$ were known, we do not observe $\tilde W_{ik}$ unlike $Z_i$. However, when $\mathcal{G}$ is Gaussian family of priors $\{\phi(\,\cdot\, |0,\tau^{1/2}\sigma): \tau>0\}$, then efficient estimates of  $R_{2,n}$ can be easily constructed. In this case the risk function reduces to terms involving second moments of $\tilde W_{ik}$ which can be easily calculated leading to quadratic functions of ${\gamma}_i$. As such, for $g=\phi(\,\cdot\, |0,\tau^{1/2}\sigma)$, the difference {$R_{1,n}[\bm \theta,\cb](g)-R_{2,n}[\bm \theta,\cb](g)$ equals: 
    \begin{align}
    \label{eq:unbiased.gaussian}
-\frac{1}{m_n}
\sum_{i=1}^{n}\sum_{k=1}^{\tilde{K}_{i}}&\Big\{
\log \frac{(\tau + u_{i}^{-2})(v_{ik}^{2}+u_{i}^{2})}{\tau(v_{ik}^{2}+u_{i}^{2})+\sigma^{2}} 
+\frac{  v_{ik}^{2} (\gamma_{i}^{2}\sigma^{-2}-\tau)}
{(\tau u_{i}^{2}+ 1)\{\tau(v_{ik}^{2}+u_{i}^{2})+1\}}
\Big\}~.
    \end{align}
It can be estimated by replacing $\gamma_i^2$ above by $Z_i(\hat{\bm{\beta}},\hat \sigma)-u_i^{-2}$. By Theorems 4.1 and 4.2 of  \citet{george2021optimal} we know that this estimate convergences  at $\sqrt{n}$ rate. For most popular families of priors we do not have a natural estimate of $R_{2,n}[\bm \theta,\cb]$. In the following section, we develop a novel method for estimating it by using surrogates of $\tilde W_{ik}$ in the predictive log-likelihood using data fission and sample reuse.  

\subsection{Estimating the fixed effects and the variance}
For our proposed methodology, we need good point estimates of $\bbeta$ and $\sigma$. We next provide a constructive procedure yielding estimators with a guaranteed rates of convergence.   

\begin{lemma}\label{lem:lemma1}
For any $\bm{\beta} \in \mathbb{R}^d$ and $\sigma >0$, under ~\hyperref[assumption2]{A2}-~\hyperref[assumption3]{A3},
there exist estimators $\hat{\bbeta}$ and $\hat{\sigma}$ such that $\hat \sigma \geq c$ \text{a.e.} and
\begin{align}
\limsup_{n \to \infty} \; l_n^{-4}\,\ex\|\hat{\bbeta} - \bbeta \|^4 
< \infty \text{ and } |\hat{\sigma}^{-1} - \sigma^{-1}| 
= O_p(l_n), 
 \label{eq:consistent}
\end{align}
where,
$l_n = \min\{ n^{-1/4}, (n \eta_n)^{-1/2} \}$ and $c$ is a positive constant. 
\end{lemma}

A detailed proof is given in the Supplementary Material. The construction proceeds in two regimes. First suppose 
\(\displaystyle \liminf_{n\to\infty} n^{1/2}\eta_n>0\). 
Then, a significant proportion of units has replicates, and so, we form contrasts that are independent of the random-effect \(\bm{\gamma}\) for each replicated unit. Let
$\mathcal{I}=\{i:1\le i\le n,\;K_i\ge 2\}$,
and define, for \(i\in\mathcal{I}\),
\[
f(i,\beta)=\{u_{i2}(Y_{i1}-\bm{x}'_{i1}\bbeta)-u_{i1}(Y_{i2}-\bm{x}'_{i2}\bbeta)\}\,(u_{i1}^2+u_{i2}^2)^{-1/2}.
\]
Set
$\hat\beta=\arg\min_{\beta}\sum_{i\in\mathcal{I}} f^2(i,\beta)$ and $\hat\sigma^2=|\mathcal{I}|^{-1}\sum_{i\in\mathcal{I}} f^2(i,\hat\beta)$.
Under Assumptions~\hyperref[assumption2]{A2}-~\hyperref[assumption3]{A3}, these estimators satisfy \eqref{eq:consistent} with rate \(|\mathcal{I}|^{-1/2}\), which (since \(|\mathcal{I}|\) is of order \(n\eta_n\)) yields the rate \((n\eta_n)^{-1/2}\).

In the complementary regime, where replicates are scarce, we employ batching. Partition the \(n\) units into \(B_n\) disjoint batches of size \(\sqrt{n}\) (so \(B_n\approx\sqrt{n}\)). Aggregate the responses within each batch and fit a (misspecified) linear model that ignores the random-effect mean to these aggregated observations; apply ordinary least squares to estimate \(\bbeta\) and \(\sigma\) on the aggregated data. The variance of the resulting estimators under the misspecified model is \(O(B_n^{-1})=O(n^{-1/2})\), hence the parameter estimation error is \(O_p(n^{-1/4})\). Moreover, since the random effects \(\{\gamma_i\}\) are i.i.d.~and satisfy Assumption~\hyperref[assumption2]{A2}, the bias induced by the misspecification is also \(O(n^{-1/4})\). As the model is Gaussian, we also get $L_4$ convergence of the fixed effects at the same rate. Thus, the batching estimator attains the rate \(n^{-1/4}\).

Combining the two constructions yields estimators \(\hat{\bm{\beta}}\) and \(\hat{\sigma}\) with convergence rate \(l_n = \min\{n^{-1/4}, (n\eta_n)^{-1/2}\}\), as stated in the result.

\subsection{Estimating predictive log-likelihood in KL risk}\label{sec.3.1}
For any fixed $\bm{\beta}$, $\sigma$, in model \eqref{eq:model.11}-\eqref{eq:model.12}, $\tilde W_{ik}$ is Gaussian with the same mean as $v_{ik}Z_{i}$ but lower variance. To construct an estimate of $R_{2,n}[\bm \theta,\cb](g)$, we plan to use estimates of $\tilde W_{ik}$ within the log-marginal likelihood that are based on a collection of $\{Z_{j}:j=1,\ldots,n\}$. Each of those estimates have the same variance as $\tilde W_{ik}$ but not necessarily the same mean. However, as $\gamma_i$'s are exchangeable, we expect the average of the log-likelihood to well-estimate the predictive log-likelihood at $\tilde W_{ik}$. 

For each $i=1,\ldots,n$ and $k=1,\ldots,\Kf_i$ define the set:
$S_{ik}=\{j \in \{1,\ldots,n\}:  u_j^2-u_{i}^2-v_{ik}^2\geq 0\}$.
For all $j \in S_{ik}$ define,  
$d_{ikj}=u_{j}^{-2}v_{ik}^{2}((v_{ik}^{2}+u_{i}^{2})^{-1}u_{j}^{2}-1)$.
By  construction note that $d_{ikj}\geq 0$, { so we can add Gaussian noise with the right scale to $v_{ik}Z_i$ to match the predictive variance target without destabilizing the construction.} 
Next, for each $i=1,\ldots,n$ and $k=1,\ldots,\Kf_i$ and $j \in S_{ik}$, we construct a new variable by adding scaled i.i.d. standard normal noise $\zeta_{ikj}$: 
$$\hat W_{ikj}=v_{ik}Z_i(\hat{\bm{\beta}},\hat \sigma) + d_{ikj}^{1/2} \hat \sigma 
\zeta_{ijk}~.$$
Note that, by construction we have:
$\text{Var}( \hat W_{ikj}) {=}\text{Var}(\tilde W_{ik}(\bbeta,\sigma))+O_p(n^{-\alpha})$ and $\ex(\hat W_{ikj}-\tilde W_{ik}(\bbeta,\sigma))=v_{ik}\red{(\gamma_{j}-\gamma_{i})}+O_p(n^{-\alpha})$. 
Let $|S_{ik}|$ denote the cardinality of the set $S_{ik}$. If $|S_{ik}|>0$, we estimate $\ex_{\gamma_i} \big\{\log \tilde m_g(\tilde W_{ik}(\bbeta,\sigma);u_i,v_{ik},\sigma)\big\}$ in \eqref{eq:risk.2.2} by 
\begin{align}\label{eq:r_ik}
\hat r_{ik}(g)=|S_{ik}|^{-1} \sum_{j\in S_{ik}}\log \tilde m_g(\hat W_{ikj};u_i,v_{ik},\hat \sigma).
\end{align}
Let $\mathcal{I}(h)
=\{(i,k): |S_{ik}|\geq h, i=1,\ldots,n \text{ and } k=1,\ldots, \Kf_i \}$. For any fixed $h \in \mathbb{N}$, we propose to estimate $R_{2,n}[\bm \theta,\cb](g)$ by 
\begin{align}\label{eq:hat.R2}
\tilde  R_{2,n}^g(h)=\bigg(\sum_{i=1}^n \Kf_i\bigg)^{-1}\sum_{i=1}^n \sum_{k=1}^{\Kf_i} \hat r_{ik}(g) \,\,1\{(i,k) \in \mathcal{I}(h)\}~.
\end{align}
In \eqref{eq:hat.R2}, \(h\) is a tuning parameter, which, unless otherwise specified, is set to \(1\). We also allow \(h\) to vary with \(n\). In Section~\ref{sec:simulation}, we provide a comparative analysis of the behavior of our risk estimator for different choices of the sequence \((h_n)_{n \geq 1}\). 

The rationale behind our proposed KL risk estimator is that though each of $\hat W_{ijk}$s used in \eqref{eq:r_ik} have different mean than $\tilde W_{ik}$, as $\bm{\gamma}$ is an exchangeable sequence, the average risk estimator with be a good estimate of the average risk provided we have enough \emph{mixing} of $Z_i(\hat{\bbeta}, \hat \sigma)$ which heuristically can be thought of is that the number of unique members of $\{Z_i(\hat{\bbeta}, \hat \sigma): 1 \leq i \leq n \}$  used in  $\hat R_{2,n}^g(h)$ is large. We prove this rigorously in Section~\ref{sec:theory}. 

Our use of sample splitting differs from that in \citet{Brown2013Poisson}, where the split variables have identical means but different variances. For predictive KL risk, \(\tilde W_{ik}\) has lower variance than \(Z_i\), so sample splitting based solely on \(Z_i\) does not yield a reliable estimate of the marginal log-likelihood. Instead of splitting \(Z_i\), we match the variance of \(v_{ik} Z_i\) to that of \(\tilde W_{ik}\) using \(d_{ikj}\) and additional noise.

Note that, $\hat W_{ikj}$ is a random quantity as it was constructed by adding scaled white noise $\zeta$. We denote its dependence on white noise by using  $\hat W_{ikj}(\zeta)$. Consider enumerating  $\ex_{\zeta}\big\{\log \tilde m_g(\hat W_{ikj}(\zeta);u_i,v_{ik},\hat \sigma)\big\}$ where the expectation of the predictive log-likelihood is taken over white noise. We finally propose the following estimator for $R_{2,n}$:
\begin{align}\label{eq:hat.R2_again}
\hat R_{2,n}^g (h_n)=\bigg(\sum_{i=1}^n \Kf_i\bigg)^{-1}\sum_{i=1}^n \sum_{k=1}^{\Kf_i} |S_{ik}|^{-1} \sum_{j\in S_{ik}}\ex_{\zeta}\big\{\log \tilde m_g(\hat W_{ikj}(\zeta);u_i,v_{ik},\hat \sigma)\big\} \,\,\mathbf{1}\{(i,k) \in \mathcal{I}(h_n)\}
\end{align}

We prescribe to use this in \eqref{eq:r_ik} and so, the resultant risk estimate is not random and also has reduced variance than before (see \citet{Brown2013Poisson} for details on this Rao-Blackwellization step). 
Plugging in these two risk component estimates from \eqref{eq:r1.hat} and \eqref{eq:hat.R2} in \eqref{eq:risk.2} we get the risk estimate $\hat R_{n}^{\,g}(h_n)$.

\section{Theory}\label{sec:theory}
We assume that the covariates follow an exchangeable design. 
Let $\{(u_i,v_i): i=1,\ldots,n\}$ be i.i.d.~with marginal densities 
$f_u$ and $f_v$. To establish the asymptotic efficiency of the proposed method, we impose the following additional assumptions.

\noindent\textbf{A4.}\label{assumption4}
The random variables $U$ and $V$ are bounded and $\inf_{a \in \mathbb{A}} {f_u(a)}/{f_v(a)} > 0,$ where $\mathbb{A}=\{a: f_v(a)>0\}$. Also, $f_u$ is continuous and strictly positive on the boundary.

\noindent\textbf{A5.}\label{assumption5}
$\mathcal{G}$ has uniformly bounded fourth moments, i.e.,
$\sup_{g \in \mathcal{G}} \int g^4(a)\, da < \infty$.

\noindent\textbf{A6.}\label{assumption6}
The sequence $\bm{\gamma} = (\gamma_1,\ldots,\gamma_n)$ is exchangeable.

\noindent\textbf{A7.}\label{assumption7}
The sequence $\eta_n$ satisfies $\liminf_{n \to \infty}  n^{q} \eta_n >0$ for some $q \in [0,1)$.

Assumption~\hyperref[assumption4]{A4} ensures past–future covariate support alignment and prevents degeneracies in the predictive log-likelihood. Assumption~\hyperref[assumption5]{A5} guarantees integrability and variance control for the Gaussian–convolution marginals, making the risk components well-defined and enabling uniform convergence. Assumption~\hyperref[assumption6]{A6} reiterates the definition of random effects while assumption~\hyperref[assumption7]{A7} regulates the proportion of replicates in the design. Under these assumptions, we now state our main result on the asymptotic behavior of our proposed method.  

\begin{theorem}\label{thm.risk.conv}
    Under assumptions~\hyperref[assumption1]{A1}-~\hyperref[assumption7]{A7}, for $h_n$ such that $\limsup_{n \to \infty} n^{-1} \eta_n^{-1} h_n = 0$ and $g \in \mathcal{G}$, the estimate $\hat R_n^{\,g}(h_n)$ of the KL predictive risk of the \prde\, $\hat{p}[\hat{\bbeta},\hat \sigma, g]$ in \eqref{eq:model.11}-\eqref{eq:model.12} satisfies: 
    $$\R_n[\bm{\theta},\cb](g)- \hat R_n^{\,g}(h_n)=O_p(c_n) \text{ as } n \to \infty,$$
    where, $c_n=\max\{n^{-(1-q)/2}\,( \log n),n^{-\alpha}\}$.
\end{theorem}

Assumption~\hyperref[assumption7]{A7} together with Lemma~\ref{lem:lemma1} implies that we have
$\alpha \ge 2^{-1}\max\{1-q,\,1/2\}$.
In particular, as $q \downarrow 0$ (corresponding to a nonvanishing fraction of units having replicates), the rate $c_n$ in the risk bounds approaches the parametric rate. The convergence rate deteriorates as $q$ increases, and consistency of the risk estimates is lost in the limit $q \uparrow 1$.
Theorem~\ref{thm.risk.conv}  does not address the highly data-scarce regime in which only finitely many units have replicates, that is, when $\limsup_{n\to\infty} n\eta_n < \infty$. This setting is discussed in Section~\ref{sec:simulation}.

Asymptotically, the tuning parameter $h_n$ can be as large as 
$o_p(n\eta_n)$, although in practice we typically consider small fixed values. In Section~\ref{sec:simulation}, we examine the empirical performance of the procedure for different choices of $h_n$.

Theorem~\ref{thm.risk.conv} establishes that the risk estimates are asymptotically pointwise consistent for any $g$ satisfying Assumption~\hyperref[assumption5]{A5}. 
With the additional assumption~\hyperref[assumption8]{A8} below on the class \(\mathcal{G}\), we next show that the proposed \prde\ is asymptotically optimal, in the sense that its risk converges to that of the oracle \prde. The first two restrictions in \hyperref[assumption8]{A8} are standard conditions on the class of priors. The third condition is a regularity assumption on the posterior mean based on the derivative of the marginal log-likelihood. It ensures that the sensitivity of the posterior mean grows at most quadratically, thereby enabling stable risk estimation.

\noindent\textbf{A8.}\label{assumption8}
{The class $\mathcal{G}$ of distributions satisfies the following:
\begin{enumerate}
\item[(i)] The
\(L^1\)-metric entropy of $\mathcal{G}$ satisfies
\[
\log N(\varepsilon,\mathcal G,\|\cdot\|_1) < C_{\mathcal{G}}\log \varepsilon^{-1},
\quad 0<\varepsilon < \mathrm{diam}(\mathcal{G}),
\]
for some $C_{\mathcal{G}}>0$,
where 
$\log N(\varepsilon,\mathcal G,\|\cdot\|_1)$ is the metric entropy of $\mathcal G$ with respect to $\|\cdot\|_1$,
and
$\mathrm{diam}(\mathcal{G})$ is the diameter of $\mathcal{G}$;
\item[(ii)] There exist constants \(M<\infty\) and \(\pi_*>0\) such that
\[
\inf_{g\in\mathcal G} g([-M,M])\ge \pi_*;
\]
\item[(iii)] For $\sigma\in(0,\infty)$, there exists $L_{\sigma}>0$ independent of $g\in\mathcal{G}$ such that for any $x\in\mathbb{R}$,
\begin{align}
\left|\frac{\partial}{\partial x}\log \int \phi(x\,;\,\gamma,\sigma) g(\gamma)d\gamma
\right|
+
\left|\frac{\partial}{\partial \sigma}\log \int \phi(x\,;\,\gamma,\sigma) g(\gamma)d\gamma
\right| 
\le L_{\sigma}(1+|x|+|x|^{2});
\label{eq: log marginal derivative}
\end{align}
\end{enumerate}
}

\begin{corollary}\label{thm.main_criteria_result}
Let $\max_{i=1,\ldots,n}|\gamma_{i}|=O_p(a_n)$. Then, 
under assumptions~\hyperref[assumption1]{A1}-\hyperref[assumption8]{A8}, the KL risk of our proposed predictor given by \eqref{eq:proposed} satisfies: 
$$\R_n[\cb](\bm{\theta},\hat{p}[\hat{\bm{\beta}},\hat{\sigma}, \hat{g}])
-\R_n[\cb](\bm{\theta},p^{\sf OR})
=O_p(c_n \,a_{n}\log n)
\text{ as } n \to \infty,$$
where, $c_n=\max\{n^{-(1-q)/2}\,( \log n),n^{-\alpha}\}$.
\end{corollary}

Corollary~\ref{thm.main_criteria_result} applies to several commonly used prior families, three of which we highlight here. 
For spike-and-slab priors, let $g(\tau)=(1-\eta)\delta_0+(\eta a/2)e^{-a|\tau|}$ with hyper-parameters $(\eta,a)\in (\underline{\eta},\overline{\eta})\times(\underline{a},\overline{a})$,
where $0<\underline{\eta}\le \overline{\eta}<1$ and $0<\underline{a}\le \overline{a}<\infty$ are prespecified bounds
; this class $\mathcal{G}_1$ is widely used in sparse estimation \citep{rovckova2018spike}. 
For scaled Gaussian mixtures, let $g(\tau)=\sum_{l=1}^L \pi_l \phi(\tau;0,\nu_l)$, where $\{\nu_l\}_{l=1}^L$ is a fixed variance grid and $(\pi_1,\ldots,\pi_L)$ are mixture weights, so that $\mathcal{G}_2$ is parameterized by the $L-1$ free proportions. 
For discrete priors, let $g(\tau)=\sum_{l=1}^L \pi_l \delta_{\tau_l}$, where $\{\tau_l\}$ is a fixed grid and $\{\pi_l\}$ are probability weights, defining the class $\mathcal{G}_3$ with $L-1$ free weight parameters. The proof that these families satisfy \hyperref[assumption8]{A8}, used in Corollary~\ref{thm.main_criteria_result}, is provided in the appendix.

\begin{proposition}\label{prop.A8}
    The above discussed three prior families satisfy Assumption~\hyperref[assumption8]{A8}.
\end{proposition}

\subsection{Proof Overview}
To understand the mechanism driving the proof, consider the dependency fraction
\[
D_n(h_n)=\kappa_n^{-1}\sum_{i=1}^n\sum_{k=1}^{\Kf_i} |S_{ik}|^{-1}\,\mathbf{1}\{|S_{ik}|\geq h_n\}.
\]
Recall that $Z_i(\hat{\bbeta},\hat{\sigma})$s are reused in constructing the risk estimator $\hat R_{2,n}^g$. Heuristically, $D_n(h_n)$ quantifies the overall fraction of reuse of the individual $Z_i$’s. $|S_{ik}|$ denotes the effective amount of historical information contributing to the $(i,k)$th component of the risk estimate. When $|S_{ik}|$ is large, each reused $Z_j(\bbeta,\sigma)$ has only minor influence. In contrast, small pooling sets imply stronger individual influence and greater dependence among the units used in the risk estimate. A key step in the proof of Theorem~\ref{thm.risk.conv} is to establish asymptotic control of $D_n(h_n)$. Smaller values of $D_n(h_n)$ correspond to more effective mixing across units and faster convergence of the risk estimator.

For additional intuition regarding the growth of $D_n(h_n)$, consider the simplified highly data-scarce setting where $h_n=1$ and the $v_{ik}$ are uniformly small, in the sense that $\sup_{i,k} v_{ik}^2 \le \inf_{1\le i\neq j\le n}|u_i^2-u_j^2|$. In this case, $D_n(1)\sim \log n$ (for derivation: see Lemma~\ref{lem:order_for_dn_no_replicate} in \hyperref[sec:appendix]{Appendix}). The proof of Theorem~\ref{thm.risk.conv} requires a more delicate asymptotic analysis to control $D_n(h)$. We bound it by $\eta_n^{-1}\log n$. In Section~\ref{sec:simulation}, we investigate the behavior of $D_n(h)$ under various designs and choices of $h_n$ via simulation.

\subsection{Highly data scarce regime}\label{sec:theory.scare}
In the highly data-scarce regime, where only finitely many units have replicates so that $\limsup_{n\to\infty} n\eta_n < \infty$, the analysis of the previous section no longer applies. In constructing the risk estimator $\hat R_{2,n}^g(h)$, we ignored all $(i,k)$ pairs not belonging to $\mathcal{I}(h)$. In this regime, however, the proportion of such pairs is non-negligible and can no longer be disregarded.

Here, we consider a modified \prde\ that combines $\hat{p}[\hat{\bm{\beta}},\hat{\sigma},g]$ with the uniform-prior-based \prde\ $\hat p^{\textsf U}$, using the latter for $(i,k)$ pairs outside $\mathcal{I}(h)$:
\[
\hat{\bm{p}}_{\textsf{m}}[\hat{\bm{\beta}},\hat{\sigma},g](\tilde{\bm{Y}}; \bm{Y},\cc)
=\prod_{(i,k)\in\mathcal{I}(h)}
\phat[\hat{\bbeta},\hat\sigma,g](\tilde Y_{ik};Z_i,u_i,v_{ik})
\prod_{(i,k)\notin\mathcal{I}(h)}
\phat[\hat{\bbeta},\hat\sigma,\textsf U](\tilde Y_{ik};Z_i,u_i,v_{ik}),
\]
where, $\phat[\hat{\bbeta},\hat\sigma,g](\tilde Y_{ik};Z_i,u_i,v_{ik})
\propto
\int
\phi(\tilde Y_{ik}-\af_{ik}^{\prime}\hat{\bbeta};\,v_{ik}\gamma,\;\hat\sigma)\;
\phi(Z_{i}[\hat{\bbeta},\hat \sigma],\gamma ,\hat \sigma) \, d\gamma$. Thus, no shrinkage is applied to coordinates outside $\mathcal{I}(h_n)$. We define the \emph{improvement factor} \[\textsf{IF}_n(h_n)=|\mathcal{I}(h_n)|/\kappa_n,\] representing the proportion of coordinates to which shrinkage is applied. The risk 
of \(\hat{\bm p}_{\textsf m}\) follows directly from the proof of Theorem~\ref{prop.bayes.risk} and is given by \eqref{eq:risk.2.0} with the terms \(R_{1,n}\) and \(R_{2,n}\) restricted to \((i,k)\in\mathcal{I}(h)\) and with $\kappa_n$ replaced by $|\mathcal{I}(h_n)|$. The risk estimator is constructed as before. The following theorem provides the corresponding asymptotic guarantee.
 
\begin{theorem}\label{thm.risk.conv_no_repicates}
Under assumptions~\hyperref[assumption1]{A1}-~\hyperref[assumption6]{A6} and if \(\limsup_{n\to\infty} n\eta_n<\infty\), for any \(g\in\mathcal{G}\) we have
\[
\R_n^{\textsf m}[\bm{\theta},\cb](g)-\hat R_n^{\,g}(h_n)=O_p(c_n),
\qquad n\to\infty,
\]
where \(c_n=\max\{n^{-1/2}(\log n),\,n^{-\alpha}\}\).
\end{theorem}

In this regime the convergence rate of the risk estimator is sharper than in the replicate-rich setting, but the improvement factor is typically much smaller. Under Assumption~\hyperref[assumption5]{A5} one has \(\textsf{IF}_n\to1\), whereas in the highly sparse regime \(\textsf{IF}_n\) may be substantially lower. In Section~\ref{sec:simulation} we study the finite-sample behavior of \(\textsf{IF}_n\) using simulations.  Below, in Theorem~\ref{thm:scarce_improvement_factor} we show that \(\textsf{IF}_n\) admits a nontrivial positive lower bound in the highly data scarce case which exceeds 25\% in symmetric designs.

\begin{theorem}\label{thm:scarce_improvement_factor}
Assume~\hyperref[assumption1]{A1}-~\hyperref[assumption6]{A6}  and $\limsup_{n\to\infty} n\eta_n < \infty$. Then, for any $\omega>0$,
\[
\liminf_{n\to\infty} \textsf{IF}_n(1) 
\;\ge\;
\omega \cdot F_{V^2}\,\!\big(F_{U^2}^{-1}(1)-F_{U^2}^{-1}(\omega)\big)
\quad \text{in P},
\]
where $F_{U^2}$ and $F_{V^2}$ denote the distribution functions of $U^2$ and $V^2$ respectively. In particular, if $F_{U^2}=F_{V^2}$ and the common distribution is symmetric, then
$\liminf_{n\to\infty} \textsf{IF}_n \ge 1/4 \text{ in P}.$
\end{theorem}
The bound shows that even without replicates, a nontrivial fraction of indices can be improved whenever future magnitudes \(V\) typically fall below a certain  gap in the past magnitudes $U$, namely \(F_{U^2}^{-1}(1)-F_{U^2}^{-1}(\omega)\). This gap captures the margin provided by the upper tail of \(U\) to accommodate future variability. The larger this gap (relative to the mass of \(V\)), the greater the improvement fraction.
\begin{lemma}\label{thm:rich_improvement_factor}
Under the assumptions of Theorem~\ref{thm.risk.conv}, as $n \to \infty$, $\textsf{IF}_{n}(h_n) \xrightarrow{P} 1 $ .
\end{lemma}
\section{Numerical Experiments}\label{sec:simulation}
Through simulation studies, we examine the finite-sample risk behavior of the proposed predictive density estimator under six distinct designs. Figures~\ref{fig:d_n_plot1} and~\ref{fig:ifactor_plot2} display the behavior of $D_n(h_n)$ (the dependency factor governing the convergence rate) and the improvement factor $\textsf{IF}_n$ as functions of $n$ and $h_n$. We consider sample sizes $n$ on the grid $\{50,100,\ldots,2500\}$ and report averages over $100$ independent replications at each grid point. Throughout, a common color scheme denotes the threshold hyperparameter $h_n$ used in the risk estimator: blue ($h_n=1$), orange ($h_n=\log n$), green ($h_n=n^{1/4}$), and red ($h_n=n^{1/2}$).

The proofs of Theorems~\ref{thm.risk.conv} and~\ref{thm:scarce_improvement_factor} show that asymptotically $D_n$ scales as $\log n$ in the highly data-scarce regime and as $\eta_n^{-1}\log n$ in the non-data scarce regime. Accordingly, the figures include in purple the growth rates predicted by the asymptotic theory. For fixed $n$, both $D_n(h)$ and $\textsf{IF}_n(h)$ decrease as $h$ increases. Figures~\ref{fig:d_n_plot1} and~\ref{fig:ifactor_plot2} illustrate this trade-off across the different designs. The reduction in $D_n$ achieved by larger thresholds comes at the expense of the improvement factor: aggressive choices of $h_n$ substantially lower $\textsf{IF}_n$, indicating that overly large thresholds are undesirable in practice. For moderate sample sizes (e.g., $n\approx 500$) and finite $h_n$, the empirical growth of $D_n(h_n)$ and $\textsf{IF}_n(h_n)$ aligns closely with the asymptotic predictions indicated by the purple reference curves. All code required to reproduce the results is available at \url{https://github.com/AbirS2026/Simulations-prde}.

\begin{figure}[!htbp]
  \centering
  \includegraphics[width=1\textwidth]{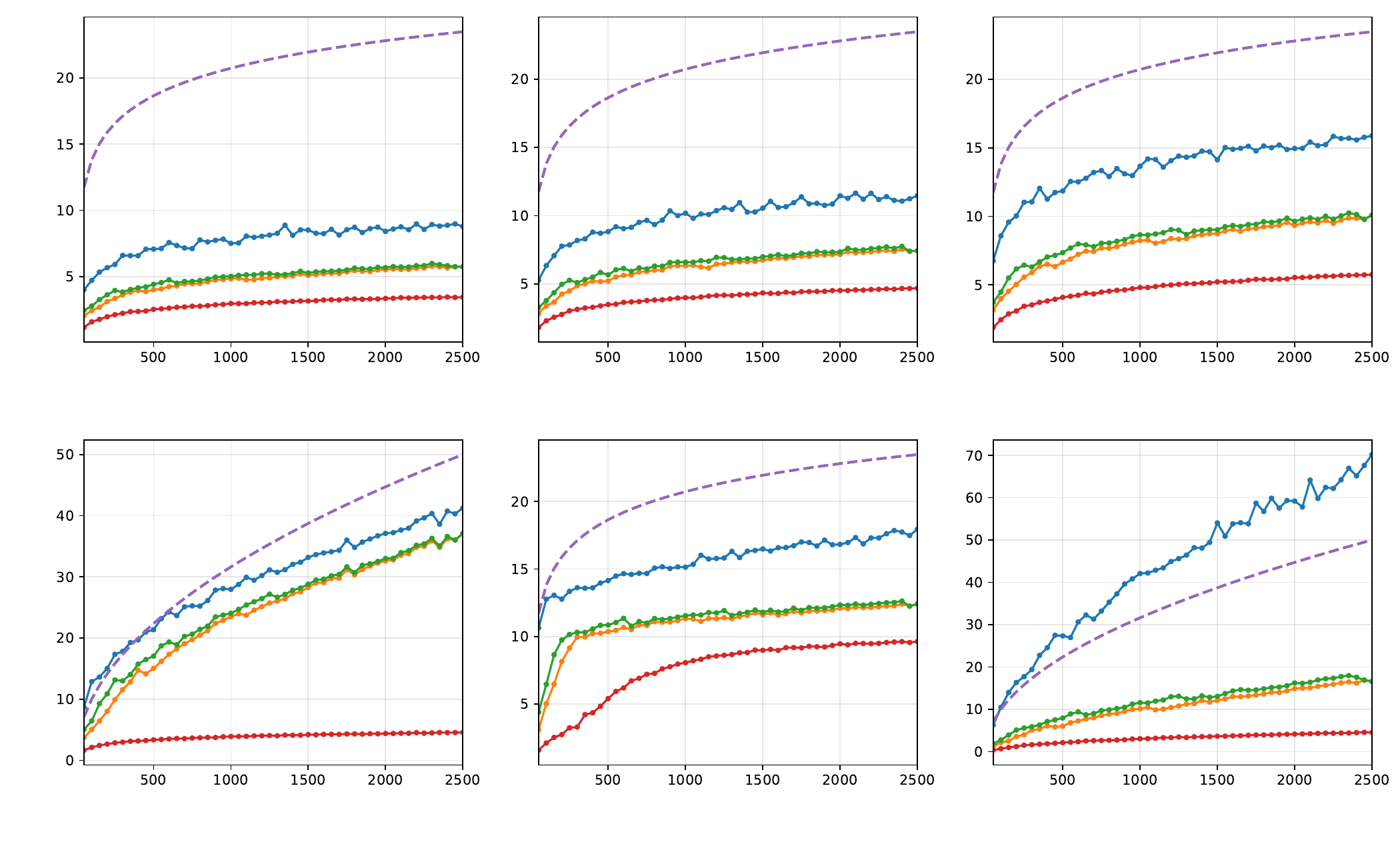}
\caption{
Plot of $D_n(h_n)$ across the six regimes A--F (ordered from top left by row) as a function of $n$. The curves correspond to different choices of $h_n$: blue ($h_n=0$), orange ($h_n=\log n$), green ($h_n=n^{1/4}$), and red ($h_n=n^{1/2}$). The dotted purple lines represent the growth rates implied by the asymptotic theory.
}
  \label{fig:d_n_plot1}
\end{figure}

\noindent We next describe the individual predictive settings and summarize the findings.

\textbf{Case A.}
We consider a no-replicate design with known $\bbeta$ and $\sigma$, corresponding to the highly data-scarce regime of Section~\ref{sec:theory.scare}. The covariates $(U,V)$ are drawn independently from $\mathcal{N}(2,1)$ truncated to $[0,4]$.

\textbf{Case B.}
This setting is identical to Case~A except that the covariates are non-Gaussian: $U,V \stackrel{\mathrm{iid}}{\sim} \mathrm{Exp}(1)$ truncated to $[0,2]$. In both Cases~A and~B, Figure~\ref{fig:d_n_plot1} shows that the empirical trajectories of $D_n$ remain uniformly below the theoretical benchmark. Figure~\ref{fig:ifactor_plot2} reports the corresponding improvement fractions, which are bounded away from zero and exceed $25\%$, consistent with the asymptotic theory. The improvement factor is noticeably larger in Case~B than in Case~A.

\textbf{Case C.}
We retain the no-replicate setting but introduce dependence between the covariates. The pair $(U,V)$ has $\mathrm{Exp}(1)$ marginals with positive dependence calibrated to achieve correlation $\rho=0.5$. The trajectories of $D_n$ and $\textsf{IF}_n$ are qualitatively similar to those in Case~B, indicating robustness to moderate dependence.

\textbf{Case D.}
We next consider a design with replication. A fraction $\eta_n=n^{-1/2}$ of units, selected uniformly at random, receive two replicated observations. The covariates satisfy $U,V\stackrel{\mathrm{iid}}{\sim}\mathcal{N}(1,1)$ truncated to $[0,2]$. The empirical growth of $D_n$ remains bounded by the theoretical $\sqrt{n}$ benchmark (purple line in the figures), while the improvement fractions are adequately high even for moderate sample sizes.

\textbf{Case E.}
We consider a design combining replication and covariate shift. We generate $U\sim \mathcal{N}(1,1)$ truncated to $[0,2]$ and $V\sim \chi^2(1)$ truncated to $[0,3]$, independently, with $\eta_n=1/10$. The resulting $D_n$ values remain well controlled, and the improvement fractions increase toward high levels as $n$ grows, in line with the asymptotic predictions.

\textbf{Case F.}
Finally, we examine a setting in which the regularity conditions of Section~\ref{sec:theory} are violated. With known $\bbeta$ and $\sigma$, we consider a no-replicate design with independent covariates, drawing $V\sim \mathrm{Unif}(1/2,1)$ and $U$ on $[1/2,1]$ with cumulative distribution function
$F(x)=1-\exp(-(1-x)^{-2}), \,x\in[1/2,1]$.
Under this specification Assumption~\hyperref[assumption4]{A4} fails. For \(h=1\), \(D_n(h)\) grows substantially faster than \(\sqrt{n}\); for larger \(h\) the growth of \(D_n(h)\) is controlled, but the improvement fractions remain comparatively small and do not approach the levels observed under the regular regimes.

\begin{figure}[!htbp]
  \centering
  \includegraphics[width=1\textwidth]{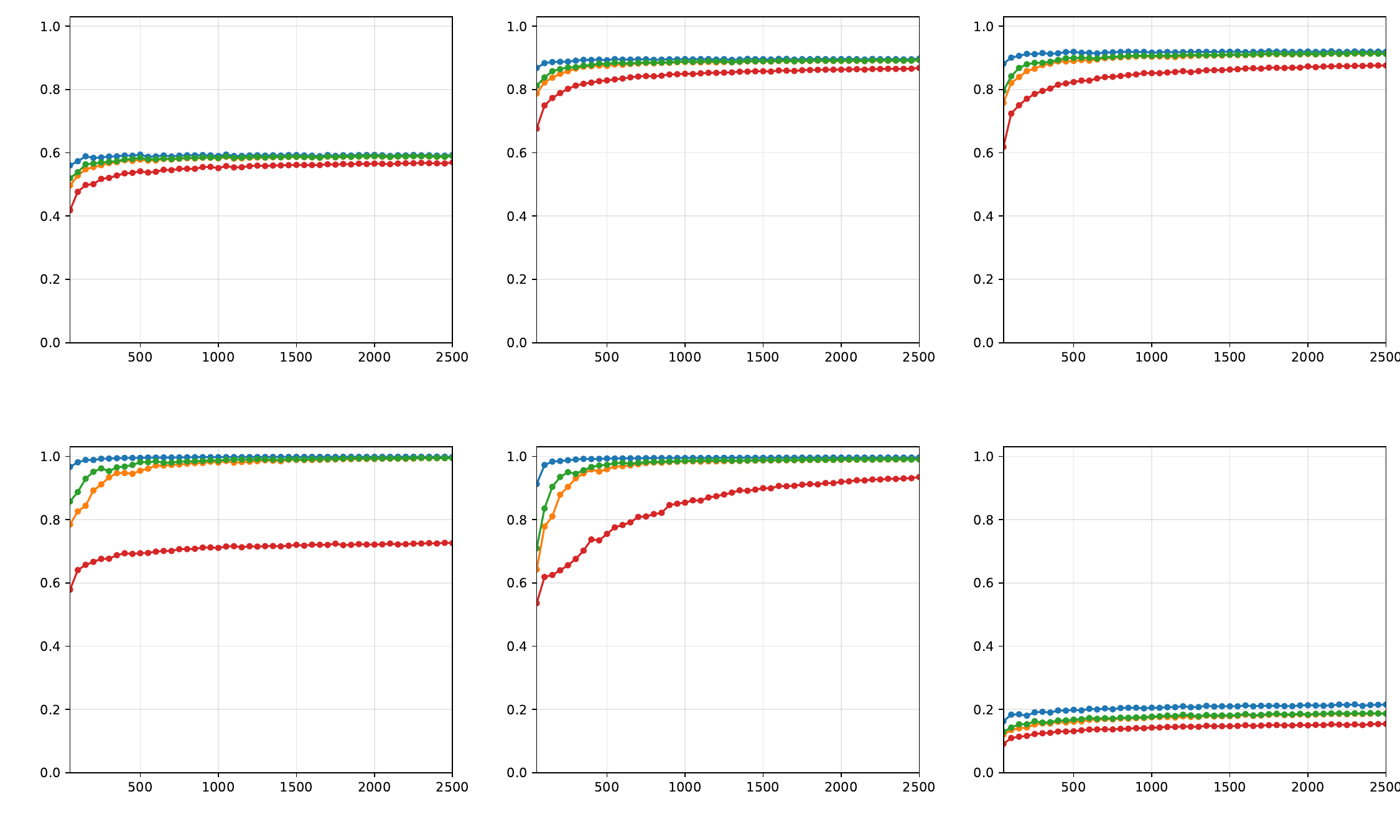}
\caption{
Plot of the improvement factor $\textsf{IF}_n$ (expressed as fractions) across the six regimes A--F (ordered from top left by row) as a function of $n$. The curves correspond to different choices of $h_n$: blue ($h_n=0$), orange ($h_n=\log n$), green ($h_n=n^{1/4}$), and red ($h_n=n^{1/2}$).
}
  \label{fig:ifactor_plot2}
\end{figure}

We next compare the KL risk of the proposed method, under the various choices of $h_n$, with competing approaches. The true random-effects distribution $g_0$ is specified as the two-component Gaussian mixture $0.7\,\mathcal N(0,0.5^{2}) + 0.3\,\mathcal N(0,1)$. We assume $\bm{\beta}$ and $\sigma$ are known and focus exclusively on the random effects.

For our proposed method, we take $\mathcal{G}$ to be the class of scaled two-component Gaussian mixtures centered at zero. We compare our approach with two competing methods: 
(a) A naive plug-in \prde\ that first computes point estimates of the random effects based on \eqref{eq:model.11} and then substitutes these estimates into \eqref{eq:model.12} for prediction. 
(b) A $g$-modeling based plug-in approach that estimates the random-effects distribution $g_0$ from the data in \eqref{eq:model.11} by fitting a member of $\mathcal{G}$ via the EM algorithm, following \citet{Efron2016DeconvolveR}. The resulting estimate $\hat g$ is then used to construct the \prde\ in \eqref{eq:model.12}.

In Figure~\ref{fig:risk_plot3}, we report the excess predictive KL risk of the competing \prdes\ relative to the true Bayes risk across the six cases considered above. Table~\ref{tab:table11} summarizes the large-sample excess KL risk (after subtracting the Bayes risk) for each method. Our method consistently achieves the lowest predictive risk across all regimes. Its excess risk converges to zero as $n$ increases. The naive plug-in estimator exhibits substantially higher risk in every setting. The $g$-modeling plug-in approach is also suboptimal and performs particularly poorly in regimes involving covariate shift.

\begin{figure}[!htbp]
  \centering
  \includegraphics[width=1\textwidth]{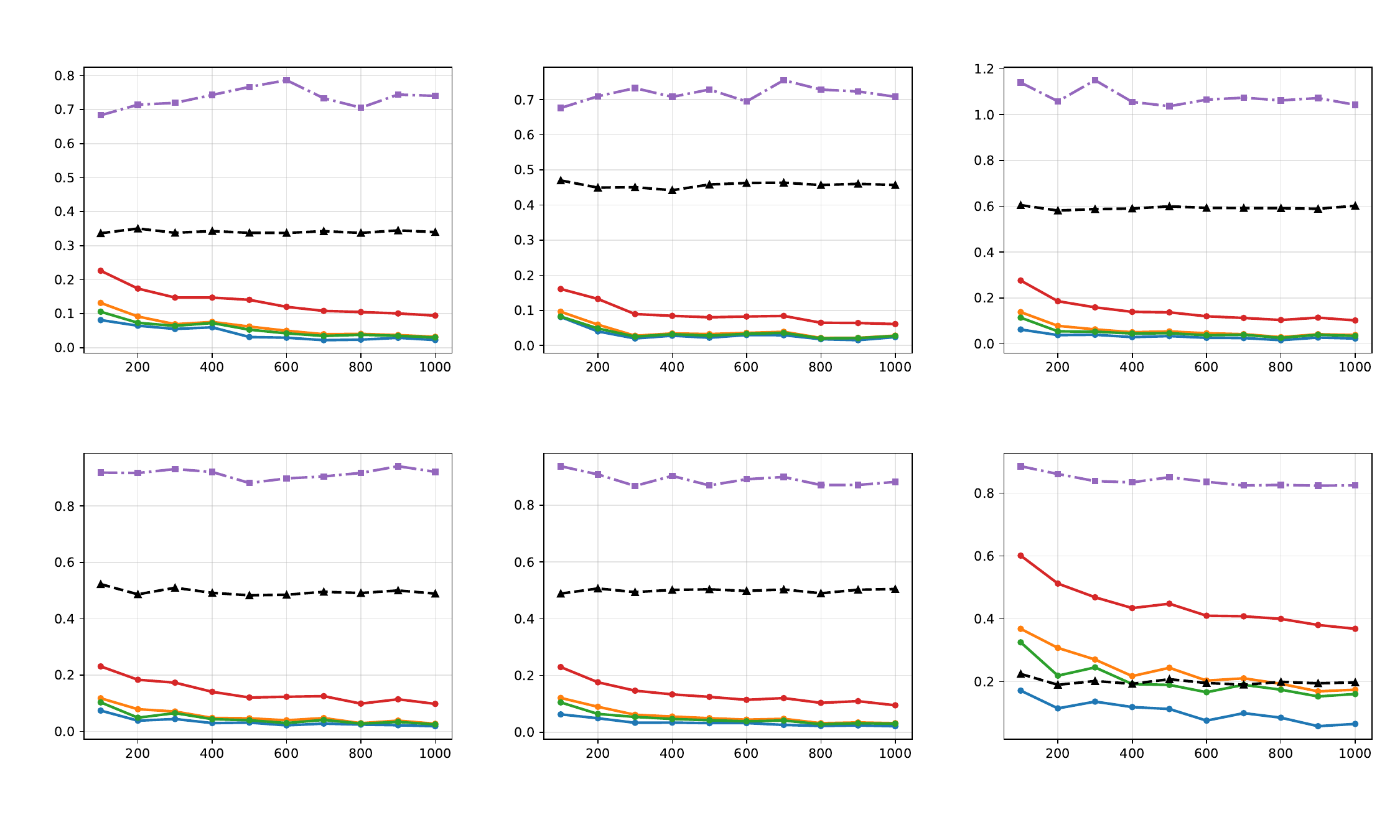}
\caption{Plot of the excess KL risk of the competing \prdes\ relative to the Bayes benchmark across the six regimes A--F (ordered from top left by row) as a function of $n$. The curves correspond to prdes for different $h_n$: blue ($h_n=0$), orange ($h_n=\log n$), green ($h_n=n^{1/4}$), and red ($h_n=n^{1/2}$), black (plugin: $g$-modeling), and purple (naive plugin).}

  \label{fig:risk_plot3}
\end{figure}

\begin{table}[t]
\centering
\renewcommand{\arraystretch}{1.15}
\begin{tabular}{lcccccc}
\hline
Method & Case 1 & Case  2 & Case  3 & Case  4 & Case  5 & Case  6 \\ \hline
Proposed: $h_n=1$           & 0.0228 & 0.0238 & 0.0230 & 0.0187 & 0.0208 & 0.0650 \\
Proposed: $h_n=\log(n)$     & 0.0323 & 0.0282 & 0.0382 & 0.0277 & 0.0313 & 0.1744 \\
Proposed: $h_n=n^{1/4}$     & 0.0304 & 0.0273 & 0.0345 & 0.0255 & 0.0295 & 0.1608 \\
Proposed: $h_n=n^{1/2}$     & 0.0946 & 0.0611 & 0.1018 & 0.0983 & 0.0947 & 0.3681 \\
Plugin: $g$-modeling        & 0.3398 & 0.4569 & 0.6021 & 0.4891 & 0.5045 & 0.1976 \\
Plugin: Naive               & 0.7402 & 0.7083 & 1.0427 & 0.9213 & 0.8821 & 0.8248 \\
\hline
\end{tabular}
\caption{Large-sample excess KL risk (at $n=1000$) of the competing \prdes.}
\label{tab:table11}
\end{table}

\section{Discussion}\label{sec:discussion}
We proposed a risk-calibrated empirical Bayes framework for predictive density estimation in high-dimensional linear mixed models under Kullback–Leibler loss. The method addresses severe data imbalance due to sparse replication and distributional shifts in future covariates. By combining flexible prior modeling with predictive KL risk estimation based on data fission and sample reuse, the resulting estimator achieves asymptotic decision-theoretic optimality and robustness to covariate shift.

Our analysis characterizes how the efficiency of risk estimation degrades as replication becomes scarce and quantifies the interplay between shrinkage, induced dependence from sample reuse, and predictive calibration. The probabilistic arguments extend the predictive heat-equation framework of \citet{george2006improved} to settings with covariates and non-Gaussian random effects.

Several directions merit further investigation. First, while we focus on KL loss, extending the calibration framework to broader divergence measures \citep{roy2026asymptotic} would provide a unified treatment of predictive optimality beyond the logarithmic score; related work in sequence models without covariates includes \citet{ghosh2018hierarchical,maruyama2016harmonic,george2021optimal,MatsudaStrawderman2021}. Second, it would be of interest to generalize the methodology to richer regression structures, including generalized linear mixed models and partially linear or functional mixed-effects models. Finally, understanding predictive calibration under more complex forms of covariate shift, including dependent or non-exchangeable designs, remains an important open problem.

\newpage

\section{Appendix}\label{sec:appendix}

\subsection{Proof of Theorem~\ref{prop.bayes.risk}}

We begin by establishing the result in the setting where $(\bbeta,\sigma)$ are known. In this case, equality in \eqref{eq:risk.2} holds exactly, without invoking asymptotic arguments. We then treat the practically relevant setting in which $(\bbeta,\sigma)$ are unknown and compare the risk of the plug-in estimator $\phat[\hat{\bbeta},\hat{\sigma},g]$ with that of the oracle $\phat[\bbeta,\sigma,g]$, which assumes knowledge of the fixed effects and variance parameter.

Let $\R_n[\bm{\theta},\cb](\hat \bbeta,\hat \sigma,g)$ denote the predictive KL risk of the \prde\ $\phat[\hat{\bbeta},\hat\sigma,g]$, and similarly let $\R_n[\bm{\theta},\cb](\bbeta,\sigma,g)$ denote the predictive KL risk of the \prde\ $\phat[\bbeta,\sigma,g]$ (i.e., the same Bayes \prde\ with $(\bbeta,\sigma)$ treated as known). For known $(\bbeta,\sigma)$, we first prove the following equality:
\begin{align}
\R_n[\bm{\theta},\cb](\bbeta,\sigma,g)
&=
a_n[\cb]
+\frac{1}{\kappa_n}\sum_{i=1}^n\sum_{k=1}^{\Kf_i}
\mathbb{E}_{\gamma_i}\!\left[\log m_g\!\Big(v_{ik}Z_i(\bbeta,\sigma);u_i,v_{ik},\sigma\Big)\right]\nonumber\\
&\quad-
\frac{1}{\kappa_n}\sum_{i=1}^n\sum_{k=1}^{\Kf_i}
\mathbb{E}_{\gamma_i}\!\left[\log \tilde m_g\!\Big(\widetilde W_{ik}(\bbeta,\sigma);u_i,v_{ik},\sigma\Big)\right].
\label{eq:known.decomp}
\end{align}

By Assumption~\hyperref[assumption1]{A1} and uniform boundedness in~\hyperref[assumption2]{A2}-~\hyperref[assumption3]{A3}, replacing $(\hat{\bbeta},\hat{\sigma})$ by $(\bbeta,\sigma)$ in the predictive density incurs an $o(1)$ contribution to $R_n[\cb](\bm{\theta},\hat p)$. Hence it suffices to work with the true $(\bbeta,\sigma)$.

Let $u_i^2:=\sum_{k=1}^{K_i}u_{ik}^2$, and define the training sufficient statistic and future residual
\[
Z_i(\bbeta,\sigma):=\frac{1}{u_i}\sum_{k=1}^{K_i} u_{ik}\,\frac{y_{ik}-x_{ik}'\bbeta}{\sigma},
\qquad
W_{ik}(\bbeta,\sigma):=\frac{\tilde y_{ik}-\tilde x_{ik}'\bbeta}{\sigma}.
\]
Conditional on $\gamma_i$, we have,
\[
v_{ik}Z_i(\bbeta,\sigma)\,\big|\,\gamma_i \sim N\!\Big(\frac{v_{ik}\gamma_i}{\sigma},\,\frac{v_{ik}^2}{u_i^2}\Big),
\qquad
W_{ik}(\bbeta,\sigma)\,\big|\,\gamma_i \sim N\!\Big(\frac{v_{ik}\gamma_i}{\sigma},\,1\Big),
\]
independently. Hence the marginal (over $g$) densities of these Gaussian mixtures are the following:
\begin{align*}
m_g\!\Big(a;u_i,v_{ik},\sigma\Big)
&=\int \phi\!\Big(a;\,\frac{v_{ik}\gamma}{\sigma},\,\frac{v_{ik}^2}{u_i^2}\Big)\,g(\gamma)\,d\gamma,\\
\tilde m_g\!\Big(a;u_i,v_{ik},\sigma\Big)
&=\int \phi\!\Big(a;\,\frac{v_{ik}\gamma}{\sigma},\,\tfrac{1}{1+v_{ik}^2 / u_i^2}\Big)\,g(\gamma)\,d\gamma,
\end{align*}
as in \eqref{m.defined.1}–\eqref{m.defined.2} (equivalently using $1+ v_{ik}^2 / u_i^2$ in the denominator).

Consider the linear transform
\[
A_{ik}:=v_{ik}Z_i(\bbeta,\sigma),
\qquad
\widetilde W_{ik}(\bbeta,\sigma):=\frac{W_{ik}(\bbeta,\sigma)-A_{ik}}{\sqrt{1+ v_{ik}^2 / u_i^2}}.
\]
The Jacobian of $(W_{ik},A_{ik})\mapsto (A_{ik},\widetilde W_{ik})$ contributes
$(1/2)\log\!\Big(1+ v_{ik}^2 / u_i^2\Big)$
to the log-density. Therefore, for each $(i,k)$,
\[
\log \frac{p\big(\tilde y_{ik}\mid\bm{\theta},\cct\big)}{\hat p\big(\tilde y_{ik}\mid\bm{Y},\cc,\cct\big)}
=
\frac12\log\!\Big(1+\frac{v_{ik}^2}{u_i^2}\Big)
+\log m_g\!\Big(A_{ik};u_i,v_{ik},\sigma\Big)
-\log \tilde m_g\!\Big(\widetilde W_{ik}(\bbeta,\sigma);u_i,v_{ik},\sigma\Big),
\]
where the estimator $\hat p$ factorizes in the transformed coordinates, using the corresponding Gaussian–mixture marginals $m_g$ and $\tilde m_g$. Averaging over $(i,k)$, integrating over $\bm{Y}$ (i.e., over $\varepsilon$ and $\tilde\varepsilon$) and $\gamma_i\sim g_0$, and normalizing by $\kappa_n=\sum_{i=1}^n\Kf_i$, we obtain
\begin{align*}
R_n[\cb](\bm{\theta},\hat p)
&=
\frac{1}{2\kappa_n}\sum_{i=1}^n\sum_{k=1}^{\Kf_i}\log\!\Big(1+\frac{v_{ik}^2}{u_i^2}\Big)
+\frac{1}{\kappa_n}\sum_{i=1}^n\sum_{k=1}^{\Kf_i}\mathbb{E}_{\gamma_i}\!\left[\log m_g\!\Big(v_{ik}Z_i(\bbeta,\sigma);u_i,v_{ik},\sigma\Big)\right]\\
&\quad-
\frac{1}{\kappa_n}\sum_{i=1}^n\sum_{k=1}^{\Kf_i}\mathbb{E}_{\gamma_i}\!\left[\log \tilde m_g\!\Big(\widetilde W_{ik}(\bbeta,\sigma);u_i,v_{ik},\sigma\Big)\right],
\end{align*}
which is precisely \eqref{eq:risk.2}–\eqref{eq:risk.2.2}. This completes the known-parameter case.\\

\noindent Similarly, define $\R_n[\bm{\theta},\cb](g):=\R_n[\bm{\theta},\cb](\hat{\bbeta},\hat\sigma,g)$, i.e., the same functional as in \eqref{eq:known.decomp} but evaluated at $(\hat{\bbeta},\hat\sigma)$. Therefore,
\begin{equation}
\R_n[\bm{\theta},\cb](g)
=
\R_n[\bm{\theta},\cb](\bbeta,\sigma,g)
+
\Big\{\R_n[\bm{\theta},\cb](\hat{\bbeta},\hat\sigma,g)-\R_n[\bm{\theta},\cb](\bbeta,\sigma,g)\Big\}.
\label{eq:split}
\end{equation}
It suffices to show that the second term in \eqref{eq:split} is $O_p(n^{-\alpha})$. To conclude this, define
\begin{align*}
R_{1,n}[\bm{\theta},\cb](\bbeta,\sigma,g)
&:=
\frac{1}{\kappa_n}\sum_{i=1}^n\sum_{k=1}^{\Kf_i}
\mathbb{E}_{\gamma_i}\!\left[\log m_g\!\Big(v_{ik}Z_i(\bbeta,\sigma);u_i,v_{ik},\sigma\Big)\right],\\
R_{2,n}[\bm{\theta},\cb](\bbeta,\sigma,g)
&:=
\frac{1}{\kappa_n}\sum_{i=1}^n\sum_{k=1}^{\Kf_i}
\mathbb{E}_{\gamma_i}\!\left[\log \tilde m_g\!\Big(\widetilde W_{ik}(\bbeta,\sigma);u_i,v_{ik},\sigma\Big)\right].
\end{align*}
Then \eqref{eq:known.decomp} yields
\[
\R_n[\bm{\theta},\cb](\bbeta,\sigma,g)=a_n[\cb]+R_{1,n}[\bm{\theta},\cb](\bbeta,\sigma,g)-R_{2,n}[\bm{\theta},\cb](\bbeta,\sigma,g),
\]
and the same identity holds with $(\bbeta,\sigma)$ replaced by $(\hat{\bbeta},\hat\sigma)$. Hence
\begin{align}
\R_n[\bm{\theta},\cb](\hat{\bbeta},\hat\sigma,g)-\R_n[\bm{\theta},\cb](\bbeta,\sigma,g)
=
\Delta_{1,n}-\Delta_{2,n},
\label{eq:Deltas}
\end{align}
where,
\begin{align*}
\Delta_{1,n}&:=R_{1,n}[\bm{\theta},\cb](\hat{\bbeta},\hat\sigma,g)-R_{1,n}[\bm{\theta},\cb](\bbeta,\sigma,g),\\
\Delta_{2,n}&:=R_{2,n}[\bm{\theta},\cb](\hat{\bbeta},\hat\sigma,g)-R_{2,n}[\bm{\theta},\cb](\bbeta,\sigma,g).
\end{align*}
Thus it is enough to prove $\Delta_{1,n}=O_p(n^{-\alpha})$ and $\Delta_{2,n}=O_p(n^{-\alpha})$.

\noindent For fixed $(i,k)$, define the two scalar maps
\begin{align*}
\Phi_{ik}(\bbeta,\sigma)
&:=
\mathbb{E}_{\gamma_i}\!\left[\log m_g\!\Big(v_{ik}Z_i(\bbeta,\sigma);u_i,v_{ik},\sigma\Big)\right],\\
\Psi_{ik}(\bbeta,\sigma)
&:=
\mathbb{E}_{\gamma_i}\!\left[\log \tilde m_g\!\Big(\widetilde W_{ik}(\bbeta,\sigma);u_i,v_{ik},\sigma\Big)\right].
\end{align*}
Then,
\begin{align*}
R_{1,n}[\bm{\theta},\cb](\bbeta,\sigma,g)&=\frac{1}{\kappa_n}\sum_{i=1}^n\sum_{k=1}^{\Kf_i}\Phi_{ik}(\bbeta,\sigma),\\
R_{2,n}[\bm{\theta},\cb](\bbeta,\sigma,g)&=\frac{1}{\kappa_n}\sum_{i=1}^n\sum_{k=1}^{\Kf_i}\Psi_{ik}(\bbeta,\sigma).
\end{align*}

\noindent We claim that under~\hyperref[assumption2]{A2} the functions $(\bbeta,\sigma)\mapsto \Phi_{ik}(\bbeta,\sigma)$ and $(\bbeta,\sigma)\mapsto \Psi_{ik}(\bbeta,\sigma)$ are continuously differentiable on sets of the form
\[
\Big\{(\bbeta,\sigma):\ \|\bbeta-\bbeta_0\|\le M,\ \underline\sigma\le\sigma\le\bar\sigma\Big\}
\]
and their first derivatives are uniformly bounded (in $(i,k)$) on such sets, for any fixed $\bbeta_0$ and $0<\underline\sigma<\bar\sigma<\infty$. Indeed, by~\hyperref[assumption3]{A3} the covariates $x,\tilde x,u,v$ are bounded; by \eqref{eq:model.11}–\eqref{eq:model.12}, $Z_i(\bbeta,\sigma)$ is affine in $\bbeta$ and in $\sigma^{-1}$, and $\widetilde W_{ik}(\bbeta,\sigma)$ is affine in $\bbeta$ and in $\sigma^{-1}$ through $Z_i(\bbeta,\sigma)$ and $\tilde Z_{ik}(\bbeta,\sigma)$; and by \eqref{m.defined.1}–\eqref{m.defined.2}, $m_g(\cdot;u,v,\sigma)$ and $\tilde m_g(\cdot;u,v,\sigma)$ are Gaussian convolutions in the first argument and in $\sigma$. Thus, differentiation under the integral in \eqref{m.defined.1}–\eqref{m.defined.2} is justified on compact $\sigma$-intervals by dominated convergence (Gaussian kernels dominate uniformly because the covariates are bounded), and yields continuous first derivatives. Moreover, the resulting derivatives of $\log m_g$ and $\log\tilde m_g$ have at most polynomial growth in their first arguments, and the latter have uniformly bounded fourth moments under~\hyperref[assumption2]{A2} since $\limsup_{n\to\infty}n^{-1}\sum_{i=1}^n \gamma_i^4<\infty$ and the noises are Gaussian.

By \hyperref[assumption1]{A1}, $\hat{\bbeta}\to\bbeta$ and $\hat\sigma^{-1}\to\sigma^{-1}$ in probability, so $\hat\sigma\to\sigma$ in probability. Hence there exist fixed constants $0<\underline\sigma<\sigma<\bar\sigma<\infty$ such that $\mathbb{P}(\hat\sigma\in[\underline\sigma,\bar\sigma])\to 1$. On the event $\{\hat\sigma\in[\underline\sigma,\bar\sigma]\}$, the multivariate mean value theorem gives, for each $(i,k)$,
\begin{align}
\Phi_{ik}(\hat{\bbeta},\hat\sigma)-\Phi_{ik}(\bbeta,\sigma)
&=
\nabla_{\bbeta}\Phi_{ik}(\bbeta^\star_{ik},\sigma^\star_{ik})^{\!\top}(\hat{\bbeta}-\bbeta)
+
\partial_{\sigma}\Phi_{ik}(\bbeta^\star_{ik},\sigma^\star_{ik})\,(\hat\sigma-\sigma),
\label{eq:MVT.Phi}\\
\Psi_{ik}(\hat{\bbeta},\hat\sigma)-\Psi_{ik}(\bbeta,\sigma)
&=
\nabla_{\bbeta}\Psi_{ik}(\tilde{\bbeta}^\star_{ik},\tilde\sigma^\star_{ik})^{\!\top}(\hat{\bbeta}-\bbeta)
+
\partial_{\sigma}\Psi_{ik}(\tilde{\bbeta}^\star_{ik},\tilde\sigma^\star_{ik})\,(\hat\sigma-\sigma),
\label{eq:MVT.Psi}
\end{align}
for some intermediate points $(\bbeta^\star_{ik},\sigma^\star_{ik})$ and $(\tilde{\bbeta}^\star_{ik},\tilde\sigma^\star_{ik})$ on the line segment between $(\bbeta,\sigma)$ and $(\hat{\bbeta},\hat\sigma)$.

By the differentiability discussion above and~\hyperref[assumption2]{A2}-\hyperref[assumption3]{A3}, there exists a (non-random) constant $C<\infty$ such that, with probability tending to one,
\begin{align}
\sup_{1\le i\le n}\sup_{1\le k\le \Kf_i}\Big\|\nabla_{\bbeta}\Phi_{ik}(\bbeta^\star_{ik},\sigma^\star_{ik})\Big\|
+
\sup_{1\le i\le n}\sup_{1\le k\le \Kf_i}\Big|\partial_{\sigma}\Phi_{ik}(\bbeta^\star_{ik},\sigma^\star_{ik})\Big|
&\le C, \label{eq:deriv.bound.Phi}\\
\sup_{1\le i\le n}\sup_{1\le k\le \Kf_i}\Big\|\nabla_{\bbeta}\Psi_{ik}(\tilde{\bbeta}^\star_{ik},\tilde\sigma^\star_{ik})\Big\|
+
\sup_{1\le i\le n}\sup_{1\le k\le \Kf_i}\Big|\partial_{\sigma}\Psi_{ik}(\tilde{\bbeta}^\star_{ik},\tilde\sigma^\star_{ik})\Big|
&\le C. \label{eq:deriv.bound.Psi}
\end{align}
Summing \eqref{eq:MVT.Phi} over $(i,k)$, dividing by $\kappa_n$, and using \eqref{eq:deriv.bound.Phi} yields, on $\{\hat\sigma\in[\underline\sigma,\bar\sigma]\}$,
\begin{align*}
|\Delta_{1,n}|
&=
\left|
\frac{1}{\kappa_n}\sum_{i=1}^n\sum_{k=1}^{\Kf_i}
\big(\Phi_{ik}(\hat{\bbeta},\hat\sigma)-\Phi_{ik}(\bbeta,\sigma)\big)
\right|\\
&\le
\frac{1}{\kappa_n}\sum_{i=1}^n\sum_{k=1}^{\Kf_i}
\left(
\Big\|\nabla_{\bbeta}\Phi_{ik}(\bbeta^\star_{ik},\sigma^\star_{ik})\Big\|\cdot\|\hat{\bbeta}-\bbeta\|
+
\Big|\partial_{\sigma}\Phi_{ik}(\bbeta^\star_{ik},\sigma^\star_{ik})\Big|\cdot|\hat\sigma-\sigma|
\right)\\
&\le C\|\hat{\bbeta}-\bbeta\| + C|\hat\sigma-\sigma|.
\end{align*}
By~\hyperref[assumption1]{A1}, $\|\hat{\bbeta}-\bbeta\|=O_p(n^{-\alpha})$. Also, since $\sigma>0$ is fixed and $\hat\sigma\to\sigma$ in probability, the map $t\mapsto t^{-1}$ is differentiable at $t=\sigma$ and
\[
\hat\sigma^{-1}-\sigma^{-1}
= -\sigma^{-2}(\hat\sigma-\sigma) + o_p(|\hat\sigma-\sigma|),
\]
so~\hyperref[assumption1]{A1} implies $|\hat\sigma-\sigma|=O_p(n^{-\alpha})$. Therefore,
\begin{equation}\label{eq:Delta1.rate}
\Delta_{1,n}=O_p(n^{-\alpha}).
\end{equation}
An identical argument using \eqref{eq:MVT.Psi} and \eqref{eq:deriv.bound.Psi} gives
\begin{equation}\label{eq:Delta2.rate}
\Delta_{2,n}=O_p(n^{-\alpha}).
\end{equation}

\noindent Combining \eqref{eq:Deltas}, \eqref{eq:Delta1.rate}, and \eqref{eq:Delta2.rate} yields
\[
\R_n[\bm{\theta},\cb](\hat{\bbeta},\hat\sigma,g)
=
\R_n[\bm{\theta},\cb](\bbeta,\sigma,g) + O_p(n^{-\alpha}).
\]
Finally, substituting \eqref{eq:known.decomp} into \eqref{eq:split} gives
\[
\R_n[\bm{\theta},\cb](g)
=
a_n[\cb]+R_{1,n}[\bm{\theta},\cb](g)-R_{2,n}[\bm{\theta},\cb](g)+O_p(n^{-\alpha}),
\]
which is \eqref{eq:risk.2}. Hence, the result follows.

\subsubsection{Proof of Lemma~\ref{lem:lemma1}}
We observe for unit \(i=1,\dots,n\) and replicate \(k=1,\dots,K_i\),
\[
y_{ik} \;=\; x_{ik}'\bbeta \;+\; u_{ik}\,\gamma_i \;+\; \sigma\,\varepsilon_{ik}, 
\qquad \varepsilon_{ik}\stackrel{iid}{\sim}N(0,1),
\]
and for future replicates,
\[
\tilde y_{ik} \;=\; \tilde x_{ik}'\bbeta \;+\; v_{ik}\,\gamma_i \;+\; \sigma\,\tilde\varepsilon_{ik}, 
\qquad \tilde\varepsilon_{ik}\stackrel{iid}{\sim}N(0,1),
\]
with the same \((\bbeta,\sigma)\) and \(\{\gamma_i\}\). We denote the proportion of replicated units by
\[
\eta_n \;:=\; \frac{1}{n}\sum_{i=1}^n \mathbf{1}\{K_i\ge 2\}\in[0,1]
\qquad\text{and}\qquad 
\mathcal{I} \;:=\; \{i:K_i\ge 2\},\quad |\mathcal{I}|=n\eta_n.
\]

\noindent We construct estimators in two regimes and then take the better rate.

\medskip
\noindent\textbf{Regime I: replicate-rich case.}
Assume
\begin{equation}\label{eq:eta-cond}
\liminf_{n\to\infty} n^{1/2}\eta_n > 0,
\end{equation}
so that \(|\mathcal{I}|=n\eta_n\) diverges at least of order \(\sqrt{n}\).
For each \(i\in\mathcal{I}\), choose two replicates \(k=1,2\) and define the contrast weights
\[
\nu_{i1} \;:=\; \frac{u_{i1}}{\sqrt{u_{i1}^2+u_{i2}^2}},
\qquad
\nu_{i2} \;:=\; \frac{u_{i2}}{\sqrt{u_{i1}^2+u_{i2}^2}},
\]
so that \(\nu_{i2}u_{i1}-\nu_{i1}u_{i2}=0\) and \(\nu_{i1}^2+\nu_{i2}^2=1\).
Define the per-unit contrast
\begin{equation}\label{eq:fi}
f_i(\bbeta) \;:=\; \nu_{i2}\big(y_{i1}-x_{i1}'\bbeta\big) \;-\; \nu_{i1}\big(y_{i2}-x_{i2}'\bbeta\big)
\;=\; s_i'(\bbeta_0-\bbeta) \;+\; \sigma\,\tilde\varepsilon_i,
\end{equation}
where \(\bbeta_0\) is the true parameter, 
\[
s_i \;:=\; \nu_{i2} x_{i1} \;-\; \nu_{i1} x_{i2},
\qquad
\tilde\varepsilon_i \;:=\; \nu_{i2}\varepsilon_{i1} \;-\; \nu_{i1}\varepsilon_{i2} \sim N(0,1),
\]

\noindent Consider the least squares estimator based on these contrasts:
\begin{equation}\label{eq:beta-hat-contrast}
\hat\bbeta \;:=\; \arg\min_{\bbeta} \sum_{i\in\mathcal{I}} f_i(\bbeta)^2.
\end{equation}
Using \eqref{eq:fi}, the first-order condition is
\[
0 \;=\; \sum_{i\in\mathcal{I}} s_i f_i(\hat\bbeta) 
\;=\; \sum_{i\in\mathcal{I}} s_i s_i'(\bbeta_0-\hat\bbeta) \;+\; \sigma\sum_{i\in\mathcal{I}} s_i \tilde\varepsilon_i,
\]
which yields the explicit linearization
\begin{equation}\label{eq:beta-linearization}
\hat\bbeta-\bbeta_0 \;=\; \sigma\,S_n^{-1} M_n,
\qquad
S_n \;:=\; \sum_{i\in\mathcal{I}} s_i s_i',\quad
M_n \;:=\; \sum_{i\in\mathcal{I}} s_i \tilde\varepsilon_i.
\end{equation}
By uniform boundedness of \(s_i\) and i.i.d. standard normal \(\tilde\varepsilon_i\), we have
\(\|M_n\|=O_p(|\mathcal{I}|^{1/2})\). Moreover, by a standard LLN and a design regularity),
\[
\frac{1}{|\mathcal{I}|} S_n \;\xrightarrow{p}\; Q \quad\text{with}\quad Q:=\mathbb{E}(s_i s_i') \succ 0,
\]
so that the smallest eigenvalue of \(S_n\) satisfies \(\lambda_{\min}(S_n)\ge c\,|\mathcal{I}|\) with probability
tending to one. Combining these bounds in \eqref{eq:beta-linearization} gives
\begin{equation}\label{eq:beta-rate-I}
\|\hat\bbeta-\bbeta_0\| \;=\; O_p\big(|\mathcal{I}|^{-1/2}\big) \;=\; O_p\big((n\eta_n)^{-1/2}\big).
\end{equation}

\noindent For \(\sigma\), define the plug-in estimator
\begin{equation}\label{eq:sigma-hat-contrast}
\hat\sigma^2 \;:=\; \frac{1}{|\mathcal{I}|}\sum_{i\in\mathcal{I}} f_i(\hat\bbeta)^2.
\end{equation}
Expanding \eqref{eq:fi} at \(\hat\bbeta\) and using \(\mathbb{E}[\tilde\varepsilon_i^2]=1\),
\[
\hat\sigma^2 - \sigma^2
\;=\; \frac{1}{|\mathcal{I}|}\sum_{i\in\mathcal{I}} \big(s_i'(\bbeta_0-\hat\bbeta)\big)^2
\;+\; \frac{2\sigma}{|\mathcal{I}|}\sum_{i\in\mathcal{I}} \tilde\varepsilon_i\,s_i'(\bbeta_0-\hat\bbeta).
\]
The first term is bounded by 
\(\|\bbeta_0-\hat\bbeta\|^2 \cdot \lambda_{\max}(S_n)/|\mathcal{I}| = O_p(|\mathcal{I}|^{-1})\),
and the second term is \(O_p(|\mathcal{I}|^{-1/2})\) by Cauchy–Schwarz and \eqref{eq:beta-rate-I}.
Thus
\begin{equation}\label{eq:sigma-rate-I}
\hat\sigma^2 - \sigma^2 \;=\; O_p\big(|\mathcal{I}|^{-1/2}\big)
\qquad\Longrightarrow\qquad
\hat\sigma^{-1}-\sigma^{-1} \;=\; O_p\big(|\mathcal{I}|^{-1/2}\big),
\end{equation}
using the delta method (or a Taylor expansion of \(x\mapsto x^{-1}\)) and \(\sigma>0\).
Combining \eqref{eq:beta-rate-I} and \eqref{eq:sigma-rate-I} yields
\[
\big\|(\hat\bbeta,\hat\sigma^{-1})-(\bbeta_0,\sigma^{-1})\big\|
\;=\; O_p\big((n\eta_n)^{-1/2}\big).
\]

\medskip
\noindent\textbf{Regime II: replicate-scarce case.}
Assume the complementary regime where \(n^{1/2}\eta_n\) is not bounded away from zero,
so that direct use of \eqref{eq:beta-hat-contrast} would be slow. We employ batching and partition the \(n\) units into \(B_n\) disjoint batches \(\{\mathcal{B}_b\}_{b=1}^{B_n}\) of equal size
\(m_n:=\lceil\sqrt{n}\rceil\), so \(B_n=\lfloor n/m_n\rfloor \asymp \sqrt{n}\). Within each batch,
form an aggregated observation by normalized averaging:
\begin{equation}\label{eq:batch-agg}
\bar y_b \;:=\; \frac{1}{\sqrt{m_n}} \sum_{i\in\mathcal{B}_b} \sum_{k=1}^{K_i} w_{ik}\,y_{ik},
\qquad
\bar x_b \;:=\; \frac{1}{\sqrt{m_n}} \sum_{i\in\mathcal{B}_b} \sum_{k=1}^{K_i} w_{ik}\,x_{ik},
\end{equation}
where the deterministic weights \(w_{ik}\) are uniformly bounded and chosen so that
\(\sum_{i\in\mathcal{B}_b}\sum_{k=1}^{K_i} w_{ik}^2\) is uniformly bounded (for normalization).
From \eqref{eq:model.11} and \eqref{eq:model.12},
\begin{equation}\label{eq:batch-model}
\bar y_b \;=\; \bar x_b'\bbeta \;+\; \bar v_b \;+\; \sigma\,\bar\varepsilon_b,
\quad\text{with}\quad
\bar v_b \;:=\; \frac{1}{\sqrt{m_n}} \sum_{i\in\mathcal{B}_b}\sum_{k=1}^{K_i} w_{ik}\,u_{ik}\,\gamma_i,
\quad
\bar\varepsilon_b \;:=\; \frac{1}{\sqrt{m_n}} \sum_{i,k} w_{ik}\,\varepsilon_{ik}.
\end{equation}
By~\hyperref[assumption2]{A2}-~\hyperref[assumption3]{A3} (boundedness of \(u_{ik}\), finite fourth moment of \(\gamma_i\)) and the normalization,
\[
\sup_b \mathbb{E}\big[\bar v_b^2\big] \;\lesssim\; \frac{1}{m_n}\sum_{i\in\mathcal{B}_b}\mathbb{E}[\gamma_i^2]
\;=\; O(1),
\qquad
\sup_b \mathbb{E}\big[\bar\varepsilon_b^2\big] \;=\; O(1).
\]
Hence, the composite batch error \(e_b:=\bar v_b+\sigma\bar\varepsilon_b\) has variance uniformly bounded
in \(b\) and \(n\). Fit the misspecified linear regression
\(\bar y_b=\bar x_b'\bbeta + e_b\) by OLS over \(b=1,\dots,B_n\):
\[
\hat\bbeta \;:=\; \arg\min_{\bbeta} \sum_{b=1}^{B_n}\big(\bar y_b - \bar x_b'\bbeta\big)^2,
\qquad
\hat\sigma^2 \;:=\; \frac{1}{B_n}\sum_{b=1}^{B_n}\big(\bar y_b - \bar x_b'\hat\bbeta\big)^2.
\]
By boundedness of \(\bar x_b\) and a standard LLN, the Gram matrix
\(B_n^{-1}\sum_{b=1}^{B_n} \bar x_b \bar x_b'\) converges in probability to a positive definite
limit (design regularity). With error variance uniformly bounded, the OLS covariance satisfies
\[
\operatorname{Var}(\hat\bbeta) \;=\; O\big(B_n^{-1}\big)
\qquad\Longrightarrow\qquad
\|\hat\bbeta-\bbeta_0\| \;=\; O_p\big(B_n^{-1/2}\big) \;=\; O_p\big(n^{-1/4}\big).
\]
The residual variance estimator \(\hat\sigma^2\) has deviation \(\hat\sigma^2-\sigma^2=O_p(B_n^{-1/2})\),
so \(\hat\sigma^{-1}-\sigma^{-1}=O_p(B_n^{-1/2})=O_p(n^{-1/4})\) by the delta method.
Potential misspecification bias originating from \(\bar v_b\) (e.g., due to a nonzero mean of \(\gamma_i\))
is of order \(O_p(m_n^{-1/2})=O_p(n^{-1/4})\) by the CLT and uniform boundedness of \(u_{ik}\), 
and hence does not alter the \(n^{-1/4}\) rate.

\noindent Therefore, we have
\[
\big\|(\hat{\bbeta}, \hat{\sigma}^{-1}) - (\bbeta, \sigma^{-1})\big\|
\;=\; O_p\!\Big(\min\{n^{-1/4},\; (n\eta_n)^{-1/2}\}\Big)
\;=\; O_p(l_n),
\]
which proves the lemma.

\subsection{Proof of Theorem~\ref{thm.risk.conv}}
    
The proof proceeds by splitting the prediction risk of the proposed estimator into two natural pieces: one that reflects how well we fit with the observed data (the training part) and one that reflects how well we forecast new outcomes (the forecasting part), and then showing that the empirical versions of both pieces concentrate tightly around their theoretical targets. For the training part, we control the variability of averages of log–mixture terms using simple convexity arguments and moment bounds, together with the consistency of the plug-in estimators for the nuisance parameters. For the forecasting part, we exploit the structure of the estimator, built by averaging over improved coordinates and their local neighborhoods, so that it is unbiased and has small variance.  A iterated-logarithm factor appears when upgrading basic variance bounds to uniform, high-probability statements.  The technical proof of the theorem involves the following steps. We need to prove:\\

\noindent \textbf{Step 1.} $
\hat R_{1,n}^{\,g}- R_{1,n}^{\,g}
= O_p\,\!\bigg(\max\Big\{n^{-1/2} ,\, n^{-\alpha}\Big\}\bigg);
$\\
\textbf{Step 2.} $\tilde R_{2,n}^{\,g}(h) -  \hat R_{2,n}^{\,g} (h)=O_p\,\!\bigg(\sqrt{\dfrac{{D_n}}{{n}} } \bigg)$;\\
\textbf{Step 3.} $\tilde R_{2,n}^{\,g}(h) -  R_{2,n}^{\,g}=O_p\!\bigg(\sqrt{\dfrac{{D_n}}{{n}}}\bigg)$.
\vspace{1cm}

\noindent\textbf{Proof of Step 1}
Let $\kappa_n:=\sum_{i=1}^n K_i$, recall
\begin{align*}
\hat R_{1,n}^{\,g}
:=\frac{1}{\kappa_n}\sum_{i=1}^n\sum_{k=1}^{K_i}
\log m_g\!\Big(v_{ik}Z_i(\hat{\bbeta},\hat{\sigma});\,u_i,v_{ik},\hat{\sigma}\Big),
\,\,\\
R_{1,n}^{\,g}
:=\frac{1}{\kappa_n}\sum_{i=1}^n\sum_{k=1}^{K_i}
\mathbb{E}_{\gamma_i}\!\Big[\log m_g\!\Big(v_{ik}Z_i(\bbeta,\sigma);\,u_i,v_{ik},\sigma\Big)\Big],
\end{align*}

\noindent By Assumption~\hyperref[assumption1]{A1}, $\hat{\bbeta}\to\bbeta$ and $\hat{\sigma}\to\sigma$ in probability and, with probability tending to one, $\hat{\sigma}\in[\underline\sigma,\bar\sigma]$ for fixed $0<\underline\sigma<\bar\sigma<\infty$. By~\hyperref[assumption3]{A3}, $\{u_i,v_{ik}\}$ are uniformly bounded and $K_i\le K_{\max}<\infty$. Write
\[
\hat R_{1,n}^{\,g}-R_{1,n}^{\,g}
=\underbrace{\Big(\hat R_{1,n}^{\,g}-\mathbb{E}[\hat R_{1,n}^{\,g}]\Big)}_{\text{sampling fluctuation}}
+\underbrace{\Big(\mathbb{E}[\hat R_{1,n}^{\,g}]-R_{1,n}^{\,g}\Big)}_{\text{plug-in bias}}.
\]

\noindent We set $\psi_{ik}:=\log m_g\!\big(v_{ik}Z_i(\hat{\bbeta},\hat{\sigma});u_i,v_{ik},\hat{\sigma}\big)$. Using that $m_g$ is a Gaussian convolution with parameters uniformly bounded by~\hyperref[assumption2]{A2} and that $Z_i(\hat{\bbeta},\hat{\sigma})$ is affine in $\gamma_i$ and Gaussian noise with bounded coefficients (~\hyperref[assumption3]{A3}),
there exists a constant $C<\infty$ such that
\begin{align}
\sup_{i,k}\,\mathbb{E}\big[\psi_{ik}^2\big]\le C
\quad\text{and}\quad
\sup_{i}\,\mathrm{Var}\!\Big(\sum_{k=1}^{K_i}\psi_{ik}\Big)\le C\,K_i^2,
\label{eq:boundvarphi}
\end{align}
with probability tending to one.

\noindent Next, we define $F_i:=\sum_{k=1}^{K_i}\psi_{ik}$; then $\{F_i\}_{i=1}^n$ are independent across $i$, and
\[
\mathrm{Var}\!\Big(\hat R_{1,n}^{\,g}\Big)
=\mathrm{Var}\!\Big(\kappa_n^{-1}\sum_{i=1}^n F_i\Big)
\le
\kappa_n^{-2}\sum_{i=1}^n \mathrm{Var}(F_i)
\le
\frac{C}{\kappa_n^{2}}\sum_{i=1}^n K_i^2
\le
\frac{C K_{\max}}{\kappa_n},
\]
with probability tending to one. Since $\kappa_n$ is of order $n$, a standard high-probability refinement of Chebyshev’s inequality yields
\[
\hat R_{1,n}^{\,g}-\mathbb{E}[\hat R_{1,n}^{\,g}]
=O_p\!\Big(n^{-1/2}\Big).
\]

\noindent by the smoothness of $(\bbeta,\sigma)\mapsto \mathbb{E}_{\gamma_i}[\log m_g(v_{ik}Z_i(\bbeta,\sigma);u_i,v_{ik},\sigma)]$ (DCT with the uniform second-moment bound above),
\[
\bigl|\mathbb{E}[\hat R_{1,n}^{\,g}]-R_{1,n}^{\,g}\bigr|
\;\le\;
C\,\mathbb{E}\bigl(\,\|\hat{\bm\beta}-\bm\beta\|+|\hat\sigma-\sigma|\,\bigr)
\;=\;O\!\left(n^{-\alpha}\right).
\]
by~\hyperref[assumption1]{A1}. \noindent Combining the two above panels gives,
\[
\hat R_{1,n}^{\,g}- R_{1,n}^{\,g}
= O_p\,\!\bigg(\max\Big\{n^{-1/2},\, n^{-\alpha}\Big\}\bigg).
\]
which completes the proof of Step~1.

\noindent \textbf{Proof of Step 2.} 
Hereafter, due to Lemma~\ref{lem:lemma1},
without any loss of generality we assume $\Kf_i=1$ for all $i=1,\ldots,n$. 
We denote $S_{ik}$ as $S_i$ as  $\Kf_i=1$. We recall the intermediate proposed estimator for $R_{2,n}$
\begin{align}\label{eq:hat.R2_new}
\tilde R_{2,n}^g(h)=\bigg(\sum_{i=1}^n \Kf_i\bigg)^{-1}\sum_{i=1}^n \sum_{k=1}^{\Kf_i} \hat r_{ik}(g) \,\,1\{(i,k) \in \mathcal{I}(h)\}~.
\end{align}

\noindent Define $s_i=|S_i|$, and recall, the set of improved coordinates, $\mathcal{I} (h)$. We use $\hat r_{ik}$ for $i\in \mathcal{I} (h)$.

\noindent We note that:
\begin{align*}
\tilde R_{2,n}^g(h) &= \bigg(\sum_{i=1}^n \Kf_i\bigg)^{-1}\sum_{i=1}^n \hat r_{ik}(g)\, \mathbf{1}\{i \in \mathcal{I} (h)\}\\ &= \bigg(\sum_{i=1}^n \Kf_i\bigg)^{-1}\sum_{i \in \mathcal{I}(h)} s_i^{-1} \sum_{j=1}^{n}\ex_{\gamma_j}\big\{\log \tilde m_g(\hat W_{ikj};u_i,v_{ik}, \hat \sigma)\big\}\cdot\mathbf{1}\{j \in S_i\}
\end{align*}
Similarly, we have,
\begin{align}\label{a.1}
\hat R_{2,n}^{\,g} (h)= \bigg(\sum_{i=1}^n \Kf_i\bigg)^{-1}\sum_{i \in \mathcal{I}(h)}\sum_{j=1}^n s_i^{-1} \ex_{\zeta}\big\{\log \tilde m_g(\hat W_{ikj}(\zeta);u_i,v_{ik},\hat \sigma)\big\}\cdot\mathbf{1}\{j \in S_i\}
\end{align}
It is easy to see that $\hat R_{2,n}^{\,g}$ is an unbiased estimate of $R_{2,n}^g(h)$. We next want to calculate its variance. For that purpose, we represent the terms in the right side of \eqref{a.1} as functions of $Z_j(\hat \beta, \hat \sigma)$ for $j=1,\ldots,n$ and apply the following upper bound
$$ \text{Var}\left(\hat R_{2,n}^{\,g} (h)\right) \leq \bigg(\sum_{i=1}^n \Kf_i\bigg)^{-1}\sum_{i \in \mathcal{I}(h)} \text{Var} \left(\sum_{j=1}^n s_i^{-1} \ex_{\zeta}\big\{\log \tilde m_g(\hat W_{ikj}(\zeta);u_i,v_{ik},\hat \sigma)\big\}\cdot\mathbf{1}\{j \in S_i\}\right)$$
Now, let's define $ \mathcal{A}_i=\sum_{j=1}^n s_i^{-1} f(v_iZ_j(\hat \beta, \hat \sigma);u_i,v_{ik}, \hat \sigma)\cdot\mathbf{1}\{j \in S_i\}.$ Using independence of the sum, 
$$ \text{Var} (\mathcal{A}_i) \leq \left(\sum_{j=1}^n s_i^{-2}\cdot\mathbf{1}\{j \in S_i\}\right) \times \{\max_{1 \leq j \leq n} \text{Var}(f(v_iZ_j(\hat \beta, \hat \sigma);u_i,v_{ik}, \hat \sigma))\}$$
Note that the second term is $O(1)$. Hence, we obtain the following bound
$$\text{Var}(\hat R_{2,n}^{\,g} (h)) \leq n^{-1}\sum_{i \in \mathcal{I}(h)} \text{Var}(\mathcal{A}_i) = O \left(n^{-1} \sum_{i \in \mathcal{I}(h)} 1/s_i\right) = O \left(\frac{D_n}{n}\right).$$  

\noindent Consequently, by Chebyshev’s inequality,
$$\hat R_{2,n}^{\,g}(h)-\mathbb{E}\big[\hat R_{2,n}^{\,g}(h)\big]
=O_p\!\bigg(\sqrt{\frac{{D_n}}{{n}} }\bigg),$$
which completes Step~2.

\noindent \textbf{Proof of Step 3.} 
By definition, 
$$R_{2,n}^g=\frac{1}{n}\sum_{i =1}^n f(\gamma_i,u_i,v_i),$$
where, 
$$f(\gamma,u,v)=\int \bigg\{\log \int \phi\bigg(\frac{\gamma-a}{\sigma (v^2r^{-1}+u^2)^{-1/2}}+z\bigg)g(a)\,da\bigg\}\,\phi(z)\,dz.$$
We assume that $\bm{\gamma}$ are i.i.d. from $\tilde g$ (this is different from $g$), $\bm{u}, \bm{v}$ are i.i.d. from $p_1$ and $p_2$ respectively. Also, $\bm{\gamma}$, $\bm{u}$, $\bm{v}$ are independent among themselves. 

Next, consider: 
$$\hat R_{2,n}^g(h)=\frac{1}{n}\sum_{i \in \mathcal{I}_{F,n}}^n \sum_{j=1}^{n}s_i^{-1}f(\gamma_j,u_i,v_i)1\{j \in S_i\}.$$
As $\bm{\gamma}$, $\bm{u}$ and $\bm{v}$ are independent, by Tower property of conditional expectation, we have: 
$$\ex \{\hat R_{2,n}^g(h)\} = \ex\{R_{2,n}^g\}=\int f(\gamma, u,v) \tilde g(\gamma) p_1(u) p_2(v) d\gamma du dv.$$
Next, note that 
$$\text{Var}(\hat R_{2,n}^g(h))=\text{Var}(E(\hat R_{2,n}^g(h)|\bm{u},\bm{v}))+\ex(\text{Var}(\hat R_{2,n}^g(h)|\bm{u},\bm{v})).$$
Following the similar steps to the proof of step 2, we obtain the variance is of $O \left( D_n / n\right)$. This completes step 3 of this proof. 

\noindent We are yet to assess the asymptotic order of $D_n$. We prove the followings lemmas.

\subsubsection{Supplementary Lemmas
}\begin{lemma}\label{lem:order_for_dn_no_replicate}
    Let $\{U_i\}_{i\ge1}$ and $\{V_i\}_{i\ge1}$ be sequences of i.i.d.\ random variables satisfying~\hyperref[assumption4]{A4}. For each $n\in\mathbb{N}$ and $i=1,\dots,n$, assume $\tilde{K_i} =1$.
Then,
\[
\mathbb{E}[D_n (1)]
= O(\log n)
\quad\text{as } n\to\infty.
\]
In particular, $D_n (1) = O_p(\log n)$.

\end{lemma}

\noindent \textbf{Proof:} Let $F_{U^2}$ and $F_{V^2}$ denote the distribution functions of $U^2$ and $V^2$ respectively, and the support of $U^2$ be $[0,M]$ for some $M >0 $. For notational simplicity, we rewrite $U^2$ and $V^2$ as $U$ and $V$, respectively, and denote the corresponding CDFs by $F_1, F_2$. The probability density function for $U$ is denoted by $f_U$. Assumption~\hyperref[assumption4]{A4} implies, there exist constants $\delta>0$ and $c_0>0$ such that, for all $x\in[0,\delta]$,
\[
\mathbb{P}(U \ge M - x)
=
1 - F_1(M-x)
\;\ge\;
c_0 x.
\]
Equivalently, the upper tail of $U$ at the right endpoint $M$ is at least linear in the distance to $M$.
Note that this is not too restrictive assumption. Since most of $U, V$ will be truncated at $M$, \[
\int_{M-\delta}^M f_U(x)\,\, dx \geq (M- (M-\delta)) \inf_{x\in(M-\delta,M]} f_U(x) \,\,  = c_0 \delta
\]
Hence, we only need $\inf_{x\in(M-\delta,M]} f_U(x) \neq 0$, which is usually true for most of the distributions unless the the tail probability is slower than a linear rate.

\noindent Since $\tilde{K_i} =1$, we can rewrite $S_{ik}$ ans $S_i$ and $s_i = |S_i|$. Hence,
\[D_n (1) = \left[\sum_{i=1}^n\frac{\mathbf{1}\{s_i>0\}}{s_i}\right]
\]
Now, we fix $n$ and $i\in\{1,\dots,n\}$, and set
\[
T_i := U_i + V_i,
\qquad
p_i := \mathbb{P}(U_j \ge T_i \mid T_i)
= 1 - F_1(T_i).
\]
For $j=1,\dots,n$ and $j \neq i$, define $I_{ij} := \mathbf{1}_{\{U_j \ge T_i\}}$. Conditional on $T_i$,
the indicators $\{I_{ij}\}_{j\neq i}$ are i.i.d.\ Bernoulli$(p_i)$, so
\[
s_i \,\big|\, T_i \sim \mathrm{Binomial}(n-1,p_i).
\]
Note that if $T_i>M$ then $p_i=0$ and $s_i=0$ almost surely, so
$\mathbf{1}_{\{s_i>0\}}/s_i = 0$ on $\{T_i>M\}$.

\medskip
\noindent
Now, note that, if $S\sim\mathrm{Binomial}(m,p)$ with $m\ge1$ and $p\in(0,1]$. Then
\[
\frac{\mathbf{1}\{S>0\}}{S} \le \frac{2}{S+1},
\]
hence,
\[
\mathbb{E}\!\left[\frac{\mathbf{1}\{S>0\}}{S}\right]
\le 2\,\mathbb{E}\!\left[\frac{1}{S+1}\right].
\]
We compute
\[
\mathbb{E}\Big[\frac{1}{S+1}\Big]
= \sum_{k=0}^n \frac{1}{k+1}\binom{m}{k} p^k (1-p)^{m-k}.
\]
Using $\frac{1}{k+1}\binom{m}{k} = \frac{1}{m+1}\binom{m+1}{k+1}$,
\[
\mathbb{E}\Big[\frac{1}{S+1}\Big]
= \frac{1}{m+1}
\sum_{k=0}^m \binom{m+1}{k+1} p^k (1-p)^{m-k}
= \frac{1}{(m+1)p}
\sum_{j=1}^{m+1} \binom{m+1}{j} p^j (1-p)^{(m+1)-j}.
\]
The sum equals $1 - (1-p)^{m+1}$, so
\[
\mathbb{E}\Big[\frac{1}{S+1}\Big]
= \frac{1 - (1-p)^{m+1}}{(m+1)p}
\le \frac{1}{(m+1)p}.
\]
Thus, for $p\in(0,1]$,
\begin{equation}\label{eq:binom-bound}
\mathbb{E}\!\left[\frac{\mathbf{1}\{S>0\}}{S}\right]
\le \frac{2}{(m+1)p}
\le 2\min\!\left\{1,\frac{1}{(m+1)p}\right\}.
\end{equation}
For $p=0$ we have $S=0$ almost surely and hence
$\mathbb{E}\big[\mathbf{1}_{\{S>0\}}/S\big]=0$, so
\eqref{eq:binom-bound} trivially holds if we interpret the right-hand side
as $2$.

\noindent Applying this with $S=s_i\mid p_i$, we obtain
\[
\mathbb{E}\!\left[\frac{\mathbf{1}\{s_i>0\}}{s_i} \,\Big|\, p_i\right]
\le 2\min\!\left\{1,\frac{1}{n p_i}\right\},
\]
and hence,
\begin{equation}\label{eq:Ei-bound}
\mathbb{E}\!\left[\frac{\mathbf{1}\{s_i>0\}}{s_i}\right]
\le
2\,\mathbb{E}\!\left[\min\!\left\{1,\frac{1}{n p_i}\right\}\right],
\end{equation}
where the integrand is interpreted as $0$ when $p_i=0$.

\noindent By symmetry, the distribution of $p_i$ does not depend on $i$, so let
$p_1$ denote a generic copy. Then
\[
\mathbb{E}[D_n]
=
\sum_{i=1}^n \mathbb{E}\!\left[\frac{\mathbf{1}\{s_i>0\}}{s_i}\right]
\le
2n\,\mathbb{E}\!\left[\min\!\left\{1,\frac{1}{n p_1}\right\}\right].
\]

\medskip
\noindent
Now, define $p_1 = p(T_1)$.
By~\hyperref[assumption4]{A4}, for $x\in[0,\delta]$,
\[
p(M-x) = \mathbb{P}(U_1 \ge M-x) \ge c_0 x.
\]
Define $p_* := p(M-\delta) = \mathbb{P}(U_1 \ge M-\delta)>0$.

\noindent If $T_1 \le M-\delta$, then $p_1 = p(T_1) \ge p_*$ and thus cannot be arbitrarily small. On the other hand, if $T_1\in(M-\delta,M]$, write $T_1=M-x$ with $x\in(0,\delta]$, so that
\[
p_1 = p(M-x) \ge c_0 x,\qquad \text{hence}\quad
p_1 \le y \ \Rightarrow\ x \le y/c_0,\ \ T_1 \ge M - y/c_0.
\]

\noindent Also note that $p_1>0$ implies $T_1\le M$, because $U_1\le M$ almost surely.
Therefore, for any $y>0$,
\[
\{0 < p_1 \le y\} \subseteq \{M - y/c_0 \le T_1 \le M\}.
\]

\noindent  Moreover, for $y>0$ small enough so that $y/c_0 \le \delta$,
\[
\mathbb{P}(0 < p_1 \le y)
\le \mathbb{P}(M - y/c_0 \le T_1 \le M)
= \int_{M-y/c_0}^{M} f_T(t)\,dt
\le C_T \frac{y}{c_0}.
\]
where $C_T =\text{sup}_{t\in [M-y/c_0, M]} \,\,f_T(t)$. Thus there exists $y_0>0$ and a constant $K>0$ such that, for all $0<y\le y_0$,
\begin{equation}\label{eq:small-p}
F_P(y) := \mathbb{P}(0 < p_1 \le y) \le K y.
\end{equation}
\noindent For $n$ large enough we have $1/n \le y_0$.
We do the following split
\[
\min\!\left\{1,\frac{1}{n p_1}\right\}
=
\mathbf{1}\{0<p_1 < 1/n\}
+ \frac{1}{n p_1}\,\mathbf{1}\{p_1 \ge 1/n\}.
\]
Hence,
\[
\mathbb{E}\!\left[\min\!\left\{1,\frac{1}{n p_1}\right\}\right]
\le
\mathbb{P}\!\left(0<p_1 < \frac{1}{n}\right)
+ \frac{1}{n}\,\mathbb{E}\!\left[\frac{1}{p_1}\,\mathbf{1}\{p_1 \ge 1/n\}\right].
\]

\noindent Using \eqref{eq:small-p},
\[
\mathbb{P}\!\left(0<p_1 < \frac{1}{n}\right)
\le K\cdot \frac{1}{n}
= O\!\left(\frac{1}{n}\right).
\]

\noindent For the second term, define the cdf $F_P$ of $p_1$ on $(0,1]$ as above.
Then,
\[
\mathbb{E}\!\left[\frac{1}{p_1}\,\mathbf{1}\{p_1 \ge 1/n\}\right]
= \int_{(1/n,1]} \frac{1}{p}\,dF_P(p).
\]
By integration by parts,
\[
\int_{1/n}^1 \frac{1}{p}\,dF_P(p)
= \left[\frac{F_P(p)}{p}\right]_{p=1/n}^{1}
+ \int_{1/n}^1 \frac{F_P(p)}{p^2}\,dp.
\]
Since $F_P(1)\le 1$ and $F_P(1/n)\le K(1/n)$, we have,
\[
\left[\frac{F_P(p)}{p}\right]_{p=1/n}^{1}
\le 1 + K.
\]

\noindent For the integral, split at some fixed $\delta'\in(0,1]$ with $\delta'\le y_0$:
\[
\int_{a/n}^1 \frac{F_P(p)}{p^2}\,dp
= \int_{a/n}^{\delta'} \frac{F_P(p)}{p^2}\,dp
+ \int_{\delta'}^{1} \frac{F_P(p)}{p^2}\,dp.
\]
For $p\in[1/n,\delta']$, \eqref{eq:small-p} gives $F_P(p)\le Kp$, hence
\[
\int_{1/n}^{\delta'} \frac{F_P(p)}{p^2}\,dp
\le K\int_{1/n}^{\delta'} \frac{1}{p}\,dp
= K\log\!\left(\delta' n\right)
= K\log n + O(1).
\]
For $p\in[\delta',1]$, we simply use $F_P(p)\le 1$:
\[
\int_{\delta'}^{1} \frac{F_P(p)}{p^2}\,dp
\le \int_{\delta'}^{1} \frac{1}{p^2}\,dp
= \frac{1}{\delta'} - 1
= O(1).
\]

\noindent Combining these,
\[
\mathbb{E}\!\left[\frac{1}{p_1}\,\mathbf{1}\{p_1 \ge 1/n\}\right]
= O(1+\log n).
\]
Therefore,
\[
\mathbb{E}\!\left[\min\!\left\{1,\frac{1}{n p_1}\right\}\right]
\le
O\!\left(\frac{1}{n}\right)
+ \frac{1}{n}\,O(1+\log n)
= O\!\left(\frac{\log n}{n}\right).
\]

\noindent Returning to \eqref{eq:Ei-bound}, we obtain
\[
\mathbb{E}\!\left[\frac{\mathbf{1}\{s_i>0\}}{s_i}\right]
\le 2\,\mathbb{E}\!\left[\min\!\left\{1,\frac{1}{n p_i}\right\}\right]
= O\!\left(\frac{\log n}{n}\right),
\]
and by symmetry this bound is the same for each $i=1,\dots,n$. Hence,
\[
\mathbb{E}[D_n (1)]
= \sum_{i=1}^n \mathbb{E}\!\left[\frac{\mathbf{1}\{s_i>0\}}{s_i}\right]
= n \cdot O\!\left(\frac{\log n}{n}\right)
= O(\log n).
\]

\noindent Finally, by Markov's inequality,
$D_n /\log n$ is bounded in probability, i.e. $D_n = O_p(\log n)$.

\noindent This completes the proof.

\noindent \textbf{Remark:} Note that, $D_n(h)$ is a non-increasing random variable in $h$. Hence, \[
\mathbb{E}[D_n (h_n)] \leq \mathbb{E}[D_n (1)] = O\,(\log n)
\]
for all tuning parameters $h_n$ satisfying the regularity conditions in Section~\ref{sec:section2}.
\begin{lemma}\label{lem:order_for_dn_with_replicates}
Under assumption~\hyperref[assumption4]{A4} and~\hyperref[assumption7]{A7}, and $\tilde{K_i} =1$ for all $i =1,2,\ldots,n$, and $\limsup_{n \to \infty} n^{-1} \eta_n^{-1} h_n = 0$
\[
\mathbb{E}\big[D_n (h_n)\big]
=\mathbb{E}\!\left[\sum_{i\in\mathcal{I} (h_n)} \frac{\mathbf{1}\{s_i>h_n\}}{s_i}\right]
=O\!\left(\frac{\log n}{\eta_n}\right)
\quad\text{as }n\to\infty,
\]
and hence $D_n=O_p\!\left(\eta_n^{-1}\log n\right)$.
\end{lemma}

\begin{proof}
For $i\in\mathcal{I}(h_n)$, define
\[
p_{1i}:=\mathbb{P}\big(U\ge U_i+V_i\big),
\qquad
p_{2i}:=\mathbb{P}\big(W\ge U_i+V_i\big),
\]
where $U$ is a single-replicate draw and $W$ is the sum of two independent replicates (both independent of $(U_i,V_i)$).
Define
\[
p_{1i} \;=\; \mathbb{P}\bigl(U > U_i + V_i\bigr),
\qquad
p_{2i} \;=\; \mathbb{P}\bigl(W > U_i + V_i\bigr),
\]
We are interested in
\[
\mathbb{E}\left[\sum_{i=1}^n \frac{1}{s_i}\,\mathbf{1}\{i \in \mathcal{I} (h_n)\}\right].
\]

\noindent Note that $i \notin S_i$ and for all $i \in \mathcal{I} (h_n)$,
\[
s_i \sim \mathrm{Bin}\bigl(n\eta_n, p_{2i}\bigr)
      + \mathrm{Bin}\bigl(n(1-\eta_n)-1, p_{1i}\bigr)
\;\;\Rightarrow\;\;
s_i \;\ge\; \mathrm{Bin}\bigl(n\eta_n, p_{2i}\bigr).
\]
Next,
\[
\mathbb{E}\!\left[\frac{1}{s_i}\,\Big|\,p_{2i}\right]
\;\le\;
\mathbb{E}\!\left[\frac{1}{n\eta_n p_{2i}}\right].
\]
Hence,
\[
\mathbb{E}\left[\sum_{i=1}^n \frac{1}{s_i}\,\mathbf{1}\{i \in \mathcal{I}_{F,n}\}\,\Big|\,p_{2i}\right]
\;\le\;
n(1-\eta_n)\,
\mathbb{E}\!\left[\min\!\left\{1,\frac{1}{n\eta_n p_{2i}}\right\}\right].
\]
Following a similar argument as in Lemma~\ref{lem:order_for_dn_no_replicate}, we conclude
\[
\mathbb{E}\!\left[\min\!\left\{1,\frac{1}{n\eta_n p_{2i}}\right\}\right]
= O\!\left(\frac{\log n}{n\eta_n}\right).
\]
Therefore, we can conclude
\[
\mathbb{E} [D_n (h_n)]
= \mathbb{E}\!\left[\sum_{i \in \mathcal{I} (h_n)} \frac{1}{s_i}\right]
= O\!\left(\frac{(1 - \eta_n)\log n}{\eta_n}\right) = O\!\left(\frac{\log n}{\eta_n}\right).
\]
This concludes the lemma.

\noindent Now, note that if $\eta_n = n^{-\beta}$, this implies
$
\mathbb{E} [D_n] = O\!\left(n^{\beta}\log n\right).
$\end{proof}
\noindent Now Theorem~\ref{thm.risk.conv} follows from aggregating the results of Step 1-3 above and Lemma~\ref{lem:order_for_dn_with_replicates}

\subsection{Proof of Theorem~\ref{thm.risk.conv_no_repicates}}

    This proof directly follows from the steps 1-3 shown above in the proof of Theorem~\ref{thm.risk.conv} along with the rate calculation of $D_n$ for highly data-scarce model in Lemma~\ref{lem:order_for_dn_no_replicate}.
\subsection{Proof of Theorem~\ref{thm:scarce_improvement_factor}}
 Let the distribution functions of $\bm{u}^2$ and $\bm{v}^2$ be denoted by $F_1$ and $F_2$, respectively. 
For notational simplicity, we rewrite $u^2$ and $v^2$ as $u$ and $v$, respectively, for
the rest of the proof. Furthermore, we assume $K_i = 1$ for all $i=1,2,\ldots,n$. Hence, $S_{ik}$ defined in
Section~\ref{sec:section2} is now considered as $S_i$. Let $U_1 \le U_2 \le \cdots \le U_n$ be the ordered sample simulated from a
continuous cdf $F_1$ with bounded support. Consider any $\omega>0$,

\noindent Define,
\begin{align*}
X_{n,\omega}
&= \frac{1}{\lfloor n\omega\rfloor}
   \sum_{i=1}^{\lfloor n\omega\rfloor}
   \mathbf{1}\!\left\{V_i>U_{\lfloor n-h_n\rfloor}-U_i\right\},\\
Y_{n,\omega}
&= \frac{1}{\lfloor n\omega\rfloor}
   \sum_{i=1}^{\lfloor n\omega\rfloor}
   \mathbf{1}\!\left\{V_i>U_{\lfloor n-h_n\rfloor}-U_{\lfloor n\omega\rfloor}\right\}.
\end{align*}
By monotonicity of the ordered sample, $U_i \le U_{\lfloor n\omega\rfloor}$ for all $i\le \lfloor n\omega\rfloor$,
hence
\begin{align}
X_{n,\omega} \le Y_{n,\omega}. \label{eq:XleY_simple}
\end{align}

\noindent Let $A_n:=U_{\lfloor n-h_n\rfloor}$ and $B_n:=U_{\lfloor n\omega\rfloor}$. By standard order-statistic
consistency,
\[
A_n \xrightarrow{P} F_1^{-1}(1),
\qquad
B_n \xrightarrow{P} F_1^{-1}(\eta).
\]
Conditional on $U$, we have
\[
\mathbb{E}\big[Y_{n,\omega}\,\big|\,U\big]
= \frac{1}{\lfloor n\omega\rfloor}
  \sum_{i=1}^{\lfloor n\omega\rfloor} \bigl(1-F_2(A_n-B_n)\bigr)
= 1-F_2(A_n-B_n).
\]
By continuity of $F_2$,
\[
\mathbb{E}[Y_{n,\omega}]
= \mathbb{E}_U\!\left[1-F_2(A_n-B_n)\right]
\xrightarrow{} 1-F_2\big(F_1^{-1}(1)-F_1^{-1}(\omega)\big).
\]
Moreover, writing $Z_i:=\mathbb{I}\{V_i>A_n-B_n\}$,
\begin{align*}
\mathrm{Var}(Y_{n,\omega})
&= \frac{1}{\lfloor n\omega\rfloor}\,\mathrm{Var}(Z_1)
 + \frac{\lfloor n\omega\rfloor-1}{\lfloor n\omega\rfloor}\,\mathrm{Cov}(Z_1,Z_2)\\
&\le \frac{1}{4\lfloor n\omega\rfloor}
 + \mathrm{Var}_U\!\left[\,1-F_2(A_n-B_n)\,\right]
 \xrightarrow{} 0,
\end{align*}
where we used conditional independence of $V_1,V_2$, boundedness of the indicators, and dominated
convergence for the $U$-variation. By Chebyshev's inequality,
\[
Y_{n,\omega}\xrightarrow{P} 1-F_2\big(F_1^{-1}(1)-F_1^{-1}(\omega)\big).
\]

\noindent From \eqref{eq:XleY_simple}, for any $a>0$,
\[
\mathbb{P}\Big(
X_{n,\omega} \le 1-F_2\big(F_1^{-1}(1)-F_1^{-1}(\omega)\big) + a
\Big)\;\xrightarrow{}\;1.
\]
Equivalently, the proportion of ``improved'' coordinates within the first $\lfloor n\omega\rfloor$ indices,
\[
1 - X_{n,\omega} \;\ge\; F_2\big(F_1^{-1}(1)-F_1^{-1}(\omega)\big) - a,
\]
with probability tending to one. Hence the overall improvement fraction satisfies
\[
\textsf{IF}_n
\;\ge\; \frac{\lfloor n\omega\rfloor}{n}\,\bigl(1-X_{n,\omega}\bigr)
\;\ge\; \omega\cdot \left(F_2\big(F_1^{-1}(1)-F_1^{-1}(\omega)\big) - a\right),
\]
and, for every $a, \omega>0$,
\[
\mathbb{P}\Big(\textsf{IF}_n > \omega\cdot
\left(F_2\big(F_1^{-1}(1)-F_1^{-1}(\omega)\big) - a\right)\;\longrightarrow\;1.
\]
Hence, the result follows. Note that, if $F_1 = F_2 = F$ is symmetric on $[0,M]$,
\[
F\big(M-F^{-1}(1/2)\big) = F(M/2) = \tfrac{1}{2}.
\]
Consequently, for $\omega=\tfrac{1}{2}$,
$\liminf_{n\to\infty} \textsf{IF}_n \ge 1/4$ in probability.

\subsubsection{Proof of Lemma~\ref{thm:rich_improvement_factor}}
 We recall the notations used in Section~\ref{sec:section2}. The sample size is $n$ and let
\[
K_i = 1 \quad \forall\, i \le \bigl\lfloor n(1-\eta_n)\bigr\rfloor,
\qquad
K_i > 1 \quad \forall\, i > \bigl\lfloor n(1-\eta_n)\bigr\rfloor.
\]
Define:
\[
U_i = u_{i1}^2 \quad \forall\, i \le \bigl\lfloor n(1-\eta_n)\bigr\rfloor,
\qquad
U_i = \sum_{j>1} u_{ij}^2 \quad \forall\, i > \bigl\lfloor n(1-\eta_n)\bigr\rfloor.
\]
We order $\{U_1, U_2,\ldots, U_{\lfloor n(1-\eta_n)\rfloor}\}$ by
$\{\tilde U_1, \tilde U_2,\ldots, \tilde U_{\lfloor n(1-\eta_n)\rfloor}\}$.\\
Similarly, we order
$\{U_{\lfloor n(1-\eta_n)\rfloor + 1}, U_{\lfloor n(1-\eta_n)\rfloor + 2},\ldots, U_n\}$ by
$\{\tilde U_{\lfloor n(1-\eta_n)\rfloor + 1}, \tilde U_{\lfloor n(1-\eta_n)\rfloor + 2},\ldots, \tilde U_n\}$.
We also denote
\[
U_{n-h_n}^* = U_{\lfloor(n-h_n)\rfloor},
\]
i.e., the $\lfloor n-h_n\rfloor$-th order statistic of the entire data, for a tuning parameter $h_n$.

In the earlier sections, we have already explained why $\{|S_i| > h_n\}$ is an event of interest.
For providing more clarity to the reader, note that this event corresponds to shrinkage for the $i$'th variable,
and we will be able to produce better shrunken risk estimates in these cases. For the $i$'th coordinate, we
would be able to provide a better shrunken risk estimate only when there exists at least one index $j$ that
satisfies $U_j \ge U_i + V_i$. In other words, we can not provide a better risk estimate for the $i$'th
coordinate if $V_i > U_{n-h_n}^* - U_i$. Define:
\begin{align*}
X_{n,\eta_n}
&= \frac{1}{\lfloor n(1-\eta_n)\rfloor}
   \sum_{i=1}^{\lfloor n(1-\eta_n)\rfloor}
   \mathbf{1}\!\left\{V_i > U_{n-h_n}^* - U_i\right\},\\
Y_{n,\eta_n}
&= \frac{1}{\lfloor n(1-\eta_n)\rfloor}
   \sum_{i=1}^{\lfloor n(1-\eta_n)\rfloor}
   \mathbf{1}\!\left\{V_i > \tilde {U}_{n-h_n} - \tilde {U}_{\lfloor n(1-\eta_n)\rfloor}\right\}.
\end{align*}
We note that,
\begin{align}
X_{n,\eta_n} \le Y_{n,\eta_n}, \label{eq:XleY}
\end{align}
since for each $i \le \lfloor n(1-\eta_n)\rfloor$ we have
$U_{n-h_n}^* - U_i \ge \tilde {U}_{n-h_n} - \tilde{U}_{\lfloor n(1-\eta_n)\rfloor}$.

The proof goes as follows:
\begin{enumerate}
\item \label{step1} $Y_{n,\eta_n} \xrightarrow{P} 0$.
\item Combining \ref{step1} and \eqref{eq:XleY} we conclude
\[
X_{n,\eta_n} \xrightarrow{P} 0,
\]
hence, for any $\delta>0$
\[
\mathbb{P}\!\left(\textsf{IF}_n > (1-\eta_n-\delta)\right)\to 1.
\]
\end{enumerate}
Define $A_n := \tilde{U}_{n-h_n}$, note that by assumption $n^{-1}h_n = o(\eta_n)$.
Hence, $\tilde{U}_{n-h_n}$ falls in the block with replicates (where $K_i \ge 2$) and
$B_n := \tilde{U}_{\lfloor n(1-\eta_n)\rfloor}$, the maximum of the no-replicate block (where $K_i = 1$).

Let the distribution functions of $u^2$ and $v^2$ be denoted by $F_1$ and $F_2$,
respectively. Note that assumption~\hyperref[assumption4]{A4} implies $F_1$ and $F_2$ have same bounded support (say $M$).
Since $F_1$ is continuous, within the no-replicate block the $U_i$'s are i.i.d. with support contained in
$[0,M]$. Therefore, as $\lfloor n(1-\eta_n)\rfloor \to \infty$,
\[
B_n \xrightarrow{P} M.
\]
Within the block with replicates, the $U_i$'s are i.i.d. sums over indices $j\ge 2$ of $u_{ij}^2$, so their
support is contained in $[0,sM]$ for some integer $s\ge 2$. Standard quantile consistency yields
\[
A_n / G_{\mathrm{right}}^{-1}(1-n^{-1}h_n) \xrightarrow{P} 1
\quad\text{and}\quad
G_{\mathrm{right}}^{-1}(1-n^{-1}h_n)\uparrow sM \quad \text{as } n \to \infty,
\]
where $G_{\text{right}}$ is the CDF of $U_j$'s in the block with replicates. so, in particular,
\[
A_n - B_n \xrightarrow{P} (s-1)M \ge M.
\]

\noindent Using the continuity of $F_2$ on $[0,M]$, we have
\[
\mathbb{E}[Y_{n,\eta_n}]
= \mathbb{E}_U\!\left[1-F_2(A_n-B_n)\right]
\xrightarrow{} 1-F_2(M)=0.
\]

\noindent We similarly control the variance. Let $Z_i := \mathbb{I}\{V_i > A_n-B_n\}$ for
$i \le \lfloor n(1-\eta_n)\rfloor$. Then
\begin{align*}
\mathrm{Var}(Y_{n,\eta_n})
&= \frac{1}{\lfloor n(1-\eta_n)\rfloor}\,\mathrm{Var}(Z_1)
 + \frac{\lfloor n(1-\eta_n)\rfloor-1}{\lfloor n(1-\eta_n)\rfloor}\,\mathrm{Cov}(Z_1,Z_2)\\
&\le \frac{1}{4\lfloor n(1-\eta_n)\rfloor}
 + \mathrm{Var}_U\!\left[1-F_2(A_n-B_n)\right],
\end{align*}
where we used conditional independence of $V_1,V_2$ given $U$ to obtain
\begin{align*}
\mathrm{Cov}(Z_1,Z_2)
&= \mathbb{E}_U\!\left[\bigl(1-F_2(A_n-B_n)\bigr)^2\right]
 - \Bigl(\mathbb{E}_U\!\left[1-F_2(A_n-B_n)\right]\Bigr)^2\\
&= \mathrm{Var}_U\!\left[1-F_2(A_n-B_n)\right].
\end{align*}
Since $A_n-B_n\xrightarrow{P} (s-1)M$ and $0\le 1-F_2(\cdot)\le 1$, dominated convergence implies
$\mathrm{Var}_U[1-F_2(A_n-B_n)]\to 0$. Hence $\mathrm{Var}(Y_{n,\eta_n})\to 0$.

\noindent By Chebyshev's inequality, $Y_{n,\eta_n}\xrightarrow{P}0$, completing Step \ref{step1} and the proof.

By Step~\ref{step1} and the inequality \eqref{eq:XleY}, we have $X_{n,\eta_n} \xrightarrow{P} 0$.
Recall that $X_{n,\eta_n}$ averages the ``no-improvement'' indicators over the first
$L_n := \lfloor n(1-\eta_n)\rfloor$ ordered statistics $\{U_i\}$, i.e., it is the proportion of non-improved
coordinates within this block. Hence the fraction of improved coordinates within the block is
$1-X_{n,\eta_n}\xrightarrow{P}1$. Consequently, the overall improvement factor satisfies
\[
\textsf{IF}_n
\;\ge\; (1-\eta_n)\,\bigl(1-X_{n,\eta_n}\bigr).
\]

\noindent We can therefore claim for any fixed $\delta>0$,
\[
\mathbb{P}\!\left(\textsf{IF}_n > 1-\eta_n-\delta\right)\;\longrightarrow\;1.
\]
This completes the proof of the lemma.

\subsubsection{Proof of Corollary~\ref{thm.main_criteria_result}}

In the proof, positive constants $b_{1},b_{2},\ldots$ do not depend on $n$ and $g$,
and notations $\lesssim$ and $\gtrsim$ indicates inequalities with multiplicative constants independent from $n$ and $g$.
From Assumption~\hyperref[assumption1]{A1},
we have $\hat{\sigma}\in[\underline{\sigma},\overline{\sigma}]$ for fixed $0<\underline{\sigma}<\overline{\sigma}<\infty$ 
and $\|\hat{\bbeta}-\bbeta\|<1$
with probability tending to one.
We denote the corresponding set of $(\hat{\bbeta},\hat{\sigma})$ by $H_{n}:=\{(\hat{\bbeta},\hat{\sigma})\,:\,
\hat{\sigma}\in[\underline{\sigma},\overline{\sigma}],\,
\|\hat{\bbeta}-\bbeta\|<1
\}$.

First, 
note that we have
\begin{align*}
&\R_n[\cb](\bm{\theta},\hat{p}[\hat{\bm{\beta}},\hat{\sigma}, \hat{g}])
-\R_n[\cb](\bm{\theta},p^{\sf OR})
\\
&=
\R_n[\cb](\bm{\theta},\hat{p}[\hat{\bm{\beta}},\hat{\sigma}, \hat{g}])
-\hat R_n^{\,g^{\sf OR}}(h_n)
+\hat R_n^{\,g^{\sf OR}}(h_n)
-\R_n[\cb](\bm{\theta},p^{\sf OR})
\\
&\le 
\R_n[\cb](\bm{\theta},\hat{p}[\hat{\bm{\beta}},\hat{\sigma}, \hat{g}])
-\hat R_n^{\,\hat{g}}(h_n)
+\hat R_n^{\,g^{\sf OR}}(h_n)
-\R_n[\cb](\bm{\theta},p^{\sf OR})
\\
&\le 2\sup_{g\in\mathcal{G}}|\widehat{\Delta}(g)| + \sup_{g\in\mathcal{G}} |\widehat{\mathcal{P}}_{n}[g]|
\end{align*}
where we define
\begin{align*}
\widehat{\Delta}_{n}[g]&:=\R_n[\bm{\theta},\cb](\bbeta,\sigma,g)- \hat R_n^{\,g}(h_n)
\\
\widehat{\mathcal{P}}_{n}[g]&:=\R_n[\cb](\bm{\theta},\hat{p}[\hat{\bm{\beta}},\hat{\sigma}, g])
-\R_n[\cb](\bm{\theta},\hat{p}[\bm{\beta},\sigma, g]),
\end{align*}
respectively.
So, we need to bound 
$\widehat{\Delta}_{n}[g]$ 
and $\widehat{\mathcal{P}}_{n}[g]$
uniformly in $g\in\mathcal{G}$.
Fix $\varepsilon>0$ arbitrarily,
and take an $\varepsilon$-net $\{g_{1},\ldots,g_{N(\varepsilon)}\}$
with $N(\varepsilon):=N(\varepsilon,\mathcal G,\|\cdot\|_1)$.
Then, we decompose $\sup_{g\in\mathcal{G}}|\widehat{\Delta}_{n}[g]|$
and $\sup_{g\in\mathcal{G}}|\widehat{\mathcal{P}}_{n}[g]|$
as 
\begin{align*}
\sup_{g\in\mathcal{G}}
\left|\widehat{\Delta}_{n}[g]\right|
&\le 
\max_{g\in\{g_{1},\ldots,g_{N(\varepsilon)}\}}
\left|\widehat{\Delta}_{n}[g]\right|
+\sup_{g,g'\in\mathcal{G}:\|g-g'\|_{1}\le \varepsilon}
\left|\widehat{\Delta}_{n}[g]-\widehat{\Delta}_{n}[g']\right|,\\
\sup_{g\in\mathcal{G}}|\widehat{\mathcal{P}}_{n}[g]|
&\le 
\max_{g\in\{g_{1},\ldots,g_{N(\varepsilon)}\}}
\left|\widehat{\mathcal{P}}_{n}[g]\right|
+\sup_{g,g'\in\mathcal{G}:\|g-g'\|_{1}\le \varepsilon}
\left|\widehat{\mathcal{P}}_{n}[g]-\widehat{\mathcal{P}}_{n}[g']\right|,
\end{align*}
respectively.

\medskip
\noindent \textbf{Step 0: Bounding the marginal densities.}
Before bounding the target quantities,
we prepare the following lemma to bound the marginal densities $m_{g}$ and $\tilde{m}_{g}$.
\begin{lemma}
Under Assumptions~\hyperref[assumption2]{A2} and~\hyperref[assumption8]{A8}(ii),
there exist positive constants $b_{1},b_{2},b_{3}$ independent from $n$ and $g$ for which we have
\begin{align}
\begin{split}
- b_{1}a^{2} + b_{2} &\le \log m_{g}\Big(a;u_i,v_{ik},\sigma\Big)
\le b_{3},\quad a\in\mathbb{R},\\
- b_{1}a^{2} + b_{2} &\le \log \tilde{m}_{g}\Big(a;u_i,v_{ik},\sigma\Big)
\le b_{3},\quad a\in\mathbb{R},
\label{eq: marginal bound}
\end{split}
\end{align}
respectively.
Further, under Assumptions
~\hyperref[assumption1]{A1},
~\hyperref[assumption2]{A2},
and~\hyperref[assumption8]{A8}(ii)
, we have
\begin{align}
\begin{split}
- b_{1}Z_{i}^{2}(\bbeta,\sigma) + b_{2} &\le \log m_{g}\Big(v_{ik}Z_{i}(\hat{\bbeta},\hat{\sigma});u_i,v_{ik},\hat{\sigma}\Big)
\le b_{3},\quad (\hat{\bbeta},\hat{\sigma})\in H_{n},\\
- b_{1}\widetilde{W}_{ik}^{2}(\bbeta,\sigma) + b_{2} &\le \log \tilde{m}_{g}\Big(\widetilde{W}_{ik}(\hat{\bbeta},\hat{\sigma});u_i,v_{ik},\hat{\sigma}\Big)
\le b_{3},\quad (\hat{\bbeta},\hat{\sigma})\in H_{n},
\label{eq: marginal bound 2}
\end{split}
\end{align}
respectively.
\end{lemma}
\begin{proof}
First, observe 
\[
\int_{[-M,M]} \phi\!\Big(a;\,\frac{v_{ik}\gamma}{\sigma},\,\frac{v_{ik}^2}{u_i^2}\Big)\,g(\gamma)\,d\gamma
\le \int \phi\!\Big(a;\,\frac{v_{ik}\gamma}{\sigma},\,\frac{v_{ik}^2}{u_i^2}\Big)\,g(\gamma)\,d\gamma
\le \frac{1}{\sqrt{2\pi (v_{ik}^{2}/u_{i}^{2})}},\quad a\in\mathbb{R}.
\]
Then, from Assumption~\hyperref[assumption2]{A2}, we have
\begin{align*}
\log m_{g}\Big(a;u_i,v_{ik},\sigma\Big)
\le {-(1/2)\log \left(2\pi \inf_{i,k}\,v_{ik}^{2}/u_{i}^{2}\right)}{:=b_{3}}.
\end{align*}
From Assumption~\hyperref[assumption8]{A8}(ii)
and Assumption~\hyperref[assumption2]{A2}, 
we have
\begin{align*}
\log m_{g}\Big(a;u_i,v_{ik},\sigma\Big)
\ge 
-\{\inf_{i,k}\,v_{ik}^{2}/u_{i}^{2}\}^{-1}a^{2}
+
{\log \pi_{*} 
- \frac{\sup_{i}\,u_{i}^{2}}{\sigma^{2}} M^{2}}
:=-b_1a^2+b_2.
\end{align*}
Summarizing these, we obtain \eqref{eq: marginal bound} for $m_{g}$.
In the same way, we obtain \eqref{eq: marginal bound} for $\tilde{m}_{g}$.

Observe, by definition of $Z_{i}(\bbeta,\sigma)$ and $\widetilde{W}_{ik}(\bbeta,\sigma)$,
\begin{align*}
Z_{i}(\hat{\bbeta},\hat{\sigma})
&=\frac{\sigma}{\hat{\sigma}}Z_{i}(\bbeta,\sigma)
-\frac{1}{u_{i}}\sum_{k=1}^{K_{i}}u_{ik}
\frac{x_{ik}'(\hat{\bbeta}-\bbeta)}{\hat{\sigma}},\\
\widetilde{W}_{ik}(\hat{\bbeta},\hat{\sigma})
&=\frac{\sigma}{\hat{\sigma}}\widetilde{W}_{ik}(\bbeta,\sigma)
-
\frac{1}{\sqrt{1+v_{ik}^{2}/u_{i}^{2}}}\left\{
\frac{1}{u_{i}}\sum_{k=1}^{K_{i}}u_{ik}
\frac{x_{ik}'(\hat{\bbeta}-\bbeta)}{\hat{\sigma}}
+\frac{\tilde{x}'_{ik}(\hat{\bbeta}-\bbeta)}{\hat{\sigma}}
\right\}
,
\end{align*}
respectively.
These, together with Assumptions~\hyperref[assumption1]{A1} and~\hyperref[assumption2]{A2} 
, yield
\begin{align}
\begin{split}
\label{eq: bounds on Z_i}
|Z_{i}(\hat{\bbeta},\hat{\sigma})| &\lesssim |Z_{i}(\bbeta,\sigma)|+1 \quad \text{for}\;(\hat{\bbeta},\hat{\sigma})\in H_{n},\\
|\widetilde{W}_{ik}(\hat{\bbeta},\hat{\sigma})|
&\lesssim |\widetilde{W}_{ik}(\bbeta,\sigma)|+1 \quad \text{for}\;(\hat{\bbeta},\hat{\sigma})\in H_{n}.
\end{split}
\end{align}
Thus combining these with \eqref{eq: marginal bound} completes the proof.
\end{proof}

\medskip
\noindent \textbf{Step 1:
Lipschitz continuity of $\widehat{\Delta}_{n}[g]$ with respect to $g$}.
We begin with bounding $\R_n[\bm{\theta},\cb](\bbeta,\sigma,g)-\R_n[\bm{\theta},\cb](\bbeta,\sigma,g')$.
Due to the singular behaviour of the log marginal density in regions where $m_{g}(a;u_{i},v_{ik},\sigma)$ is close to zero, we need to regularize $m_{g}$ as in \cite{jiang2009general}.

\noindent 
\textbf{Step 1 (i): bounding the expectation part.}
For $\delta>0$,
define
\[R_{1,n}^{\delta}[\bm \theta,\cb](\bbeta,\sigma,g)=\frac 1 {\kappa_n}\sum_{i=1}^n \sum_{k=1}^{\Kf_i} \ex_{\gamma_i}\big\{\log \max\{m_g(v_{ik}Z_i(\bm \beta,\sigma);u_i,v_{ik},\sigma),\delta\}\big\},
\]
and 
\begin{align*}
\begin{split}
&\mathcal{T}_{1,\delta}(\bm{\beta},\sigma)
:=
\sup_{g\in\mathcal G}
\frac{1}{\kappa_n}\sum_{i=1}^n\sum_{k=1}^{\Kf_i}
\mathbb \ex_{\gamma_{i}}\!\Big{[}
|\log m_g( v_{ik}Z_{i}(\bm{\beta},\sigma);u_{i},v_{ik},\sigma)|\\
&\qquad\qquad\qquad\qquad\qquad\qquad\qquad\qquad
\cdot\mathbf{1}\{m_g(v_{ik}Z_{i}(\bm{\beta},\sigma);u_{i},v_{ik},\sigma)\le \delta\}
\Big{]},
\end{split}
\end{align*}
respectively.
From the Lipschitz continuity of $\log \max\{\cdot,\delta\}$ and \eqref{eq: marginal bound}, we have
\begin{align*}
    |R^{\delta}_{1,n}[\bm{\theta},\mathbb{C}](\bbeta,\sigma,g)-R^{\delta}_{1,n}[\bm{\theta},\mathbb{C}](\bbeta,\sigma,g')|
    &\le 
    \frac{1}{\delta}
    \sup_{a,u_{i},v_{ik},\sigma}|m_g(a;u_i,v_{ik},\sigma)
    -m_{g'}(a;u_i,v_{ik},\sigma)|
    \\
    &\le
    \frac{b_{3}}{\delta}\;\|g-g'\|_{1}.
\end{align*}
By definition of $\mathcal{T}_{1,\delta}(\bm{\beta},\sigma)$, we have
\begin{align*}
    \sup_{g\in\mathcal{G}}|R^{\delta}_{1,n}[\bm{\theta},\mathbb{C}](\bbeta,\sigma,g)-R_{1,n}[\bm{\theta},\mathbb{C}](\bbeta,\sigma,g)|
    &\le \mathcal{T}_{1,\delta}(\bm{\beta},\sigma).
\end{align*}

We proceed to bound $\mathcal{T}_{1,\delta}(\bm{\beta},\sigma)$.
Inequality \eqref{eq: marginal bound} gives
\[
\mathcal{T}_{\delta}(\bm{\beta},\sigma)
\le \sup_{g\in\mathcal{G}}\frac{1}{\kappa_n}\sum_{i=1}^n\sum_{k=1}^{\Kf_i}
\mathbb{E}_{\gamma_{i}}[\{b_{1} v_{ik}^{2}Z^{2}_{i}(\bm{\beta},\sigma) +b_{2}\}
\cdot\mathbf{1}\{b_{1} v_{ik}^{2}Z^{2}_{i}(\bm{\beta},\sigma) \ge \log (1/\delta) \}
].
\]
From the conditional Gaussianity
$
v_{ik}Z_i(\bbeta,\sigma)\,\big|\,\gamma_i \sim N\!\Big(\frac{v_{ik}\gamma_i}{\sigma},\,\frac{v_{ik}^2}{u_i^2}\Big),
$
we have
\begin{align*}
&\mathbb{E}_{\gamma_{i}}[\{b_{1} v_{ik}^{2}Z^{2}_{i}(\bm{\beta},\sigma) +b_{2}\}
\cdot\mathbf{1}\{b_{1} v_{ik}^{2}Z^{2}_{i}(\bm{\beta},\sigma) \ge \log (1/\delta)\}
]\\
&\lesssim 
\int \left\{
\frac{v_{ik}^{2}}{\sigma^{2}}\gamma_{i}^{2} + a^{2}
\right\}\phi\left(a; 0, \frac{v_{ik}^{2}}{u_{i}^{2}}\right)
\cdot\mathbf{1}\{|a| \ge 
(v_{ik}/\sigma)|\gamma_{i}|
+
b_{1}^{-1/2}\sqrt{\log (1/\delta)}\}
da.
\end{align*}
Here, by definition of $a_{n}$ and from Assumption~\hyperref[assumption2]{A2},
taking a sufficiently small $\delta>0$ gives
\[
\mathbf{1}\{|a| \ge 
(v_{ik}/\sigma)|\gamma_{i}|
+
b_{1}^{-1/2}\sqrt{\log (1/\delta)}\}
\le 
\mathbf{1}\{|a| 
\ge b_{4}
(a_{n}+\sqrt{\log (1/\delta)})\}
\quad\text{for}\;\text{some}\; b_{4}>0,
\]
where the asymptotic relation $\int a^{2}\phi(a;0,1)\mathbf{1}\{|a|>M\}da \sim M\phi(M;0,1)$ as $M\to\infty$ is used.
This gives 
\begin{align*}
&\mathbb{E}_{\gamma_{i}}[\{b_{1} v_{ik}^{2}Z^{2}_{i}(\bm{\beta},\sigma) +b_{2}\}
\cdot\mathbf{1}\{b_{1} v_{ik}^{2}Z^{2}_{i}(\bm{\beta},\sigma) \ge \log (1/\delta)\}
]\\
&\lesssim 
\int 
a^{2}\phi(a; 0, 1)
\cdot\mathbf{1}\{|a| \ge b_{4}(a_{n}+\sqrt{\log (1/\delta))}\}
da
\quad\lesssim
\delta^{b_{4}^{2}/2}
(a_{n}+\sqrt{\log(1/\delta)}),
\end{align*}
where the Gaussian integral is used to obtain the second inequality.
Thus we have
\begin{align*}
\mathcal{T}_{1,\delta}(\bm{\beta},\sigma)
\lesssim \delta^{b_{4}^{2}/2} (a_{n}+\sqrt{\log (1/\delta)}).
\end{align*}
Also, define
\[
R_{2,n}^{\delta}[\bm{\theta},\cb](\bbeta,\sigma,g)
:=
\bigg(\sum_{i=1}^n \Kf_i\bigg)^{-1}
\mathbb{E}_{\gamma_i}\!\left[\log 
\max\left\{\tilde m_g\!\Big(\widetilde W_{ik}(\bbeta,\sigma);u_i,v_{ik},\sigma\Big),\delta
\right\}
\right]
\]
and 
\begin{align*}
\begin{split}
\mathcal{T}_{2,\delta}(\bm{\beta},\sigma)
&:=
\sup_{g\in\mathcal G}
\bigg(\sum_{i=1}^n \Kf_i\bigg)^{-1}\sum_{i=1}^n\sum_{k=1}^{\Kf_i}
\mathbb{E}_{\gamma_i}\!\Big{[}|\log 
\tilde m_g\!\big(\widetilde W_{ik}(\bbeta,\sigma);u_i,v_{ik},\sigma\big)|\\
&\qquad\qquad\qquad\qquad\qquad\qquad\qquad\qquad
 \cdot\mathbf{1}\{m_g\!(\widetilde W_{ik}(\bbeta,\sigma);u_i,v_{ik},\sigma)\le \delta\}\Big{]},
\end{split}
\end{align*}
respectively. From \eqref{eq: marginal bound} and the Lipschitz continuity of $\log \max\{\cdot,\delta\}$, we have
\begin{align*}
    |R^{\delta}_{2,n}[\bm{\theta},\mathbb{C}](\bbeta,\sigma,g)-R^{\delta}_{2,n}[\bm{\theta},\mathbb{C}](\bbeta,\sigma,g')|
    &\le 
    \frac{b_{3}}{\delta}\;\|g-g'\|_{1}.
\end{align*}
Since $\widetilde W_{ik}(\bbeta,\sigma)$ is conditional Gaussian, 
the same procedure as in bounding $\mathcal{T}_{1,\delta}$ gives
\[
\mathcal{T}_{2,\delta}(\bm{\beta},\sigma,u_{i},v_{ik})
\lesssim \delta^{b_{5}^{2}/2} (a_{n}+\sqrt{\log (1/\delta)})\quad\text{for some}\; b_{5}>0.
\]
Thus we obtain
\begin{align}
    &\sup_{g,g'\in\mathcal{G}:\|g-g'\|_{1}\le \varepsilon}
\left|\R_n[\bm{\theta},\cb](\bbeta,\sigma,g)-\R_n[\bm{\theta},\cb](\bbeta,\sigma,g')\right|\nonumber\\
&\le \sup_{g,g'\in\mathcal{G}:\|g-g'\|_{1}\le \varepsilon}
\left|R_{1,n}[\bm{\theta},\mathbb{C}](\bbeta,\sigma,g)-R_{1,n}[\bm{\theta},\mathbb{C}](\bbeta,\sigma,g')\right|\nonumber\\
&\qquad\qquad +
\sup_{g,g'\in\mathcal{G}:\|g-g'\|_{1}\le \varepsilon}
\left|R_{2,n}[\bm{\theta},\mathbb{C}](\bbeta,\sigma,g)-R_{2,n}[\bm{\theta},\mathbb{C}](\bbeta,\sigma,g')\right|
\nonumber\\
&\lesssim \left\{\frac{\varepsilon}{\delta} + \delta^{\min\{b_{4}^{2},b_{5}^{2}\}/2} 
\max\{a_{n},\sqrt{\log (1/\delta)}\}\right\}.
    \label{eq: diff R bound}
\end{align}

\noindent
\textbf{Step 1 (ii): bounding the empirical part.}
We proceed to bound $\hat R_n^{\,g}(h_n)-\hat R_n^{\,g'}(h_n)$.
Consider $\sup_{g,g'\in\mathcal{G}:\|g-g'\|_{1}\le \varepsilon}
|\hat{R}_{1,n}^{\,g}-\hat{R}_{1,n}^{\,g'}|$.
For $\delta>0$, 
define
\[\hat R_{1,n}^{\,g,\delta}=\frac 1 {\kappa_n}\sum_{i=1}^n \sum_{k=1}^{\Kf_i} \big\{\log \max\{ m_g( v_{ik}Z_{i}(\hat{\bbeta},\hat{\sigma});u_{i},v_{ik},\hat{\sigma}),\delta\}\big\}
\]
and 
\begin{align*}
\begin{split}
\widehat{\mathcal{T}}_{1,\delta}
:=
\sup_{g\in\mathcal G}
\frac{1}{\kappa_n}\sum_{i=1}^n\sum_{k=1}^{\Kf_i}
\left|\log m_g( v_{ik}Z_{i}(\hat{\bbeta},\hat{\sigma});u_{i},v_{ik},\hat{\sigma})\right|
\cdot\mathbf{1}\{ m_g( v_{ik}Z_{i}(\hat{\bbeta},\hat{\sigma});u_{i},v_{ik},\hat{\sigma})\le \delta\},
\end{split}
\end{align*}
respectively.
Then, we have
\begin{align*}
    |\hat R_{1,n}^{\,g,\delta}-\hat R_{1,n}^{\,g',\delta}|
    \lesssim \frac{1}{\delta}\;\|g-g'\|_{1},
    \quad \text{and} \quad
    |\hat R_{1,n}^{\,g,\delta}-\hat R_{1,n}^{\,g}|
    \le \widehat{\mathcal{T}}_{1,\delta}.
\end{align*}

Here we bound $\widehat{\mathcal{T}}_{1,\delta}$.
Recall Inequality \eqref{eq: bounds on Z_i}:
\[
|Z_{i}(\hat{\bbeta},\hat{\sigma})|
\lesssim 
|Z_{i}(\bbeta,\sigma)|+1 \quad \text{for}\;(\hat{\bbeta},\hat{\sigma})\in H_{n}.
\]
Then, we have
\begin{align*}
    \widehat{\mathcal{T}}_{1,\delta} 
    &\lesssim 
    \frac{1}{\kappa_n}\sum_{i=1}^n\sum_{k=1}^{\Kf_i}
(
Z_{i}^{2}(\bbeta,\sigma)+1
)
\cdot\mathbf{1}\{ |Z_{i}(\bm{\beta},\sigma)| \gtrsim
a_{n}+\sqrt{\log (1/\delta)}\}\\
&\lesssim \frac{1}{\kappa_n}\sum_{i=1}^n\sum_{k=1}^{\Kf_i}
\ex_{\gamma_{i}}
\left[
Z_{i}^{2}(\bbeta,\sigma)+1)
\cdot\mathbf{1}\{ |Z_{i}(\bm{\beta},\sigma)| \gtrsim 
a_{n}+\sqrt{\log (1/\delta)}\}
\right] \sqrt{\log n}
\end{align*}
for $(\hat{\bbeta},\hat{\sigma})\in H_{n}$, where the conditional Gaussianity of $Z_{i}(\bbeta,\sigma)$ 
and the $\chi^2$-tail bound \citep{laurent2000adaptive}
are employed to obtain the latter inequality.
So, noting that $\kappa_{n}$ is of order $n$,
and taking a sufficiently small $\delta>0$,
we conclude 
\begin{align*}
\widehat{\mathcal{T}}_{1,\delta}
\lesssim \delta^{b_{6}} 
\max\{a_{n},\sqrt{\log (1/\delta)}\}
\sqrt{\log n}
\quad\text{for}\;\text{some}\; b_{6}>0
\end{align*}
with probability tending to one.
Therefore we obtain
\begin{align}
    \sup_{g,g'\in\mathcal{G}:\|g-g'\|_{1}\le \varepsilon}
\left|\hat{R}_{1,n}^{\,g}-\hat{R}_{1,n}^{\,g'}\right|
\lesssim \left\{\frac{\varepsilon}{\delta} + \delta^{b_{6}} \sqrt{
\max\{a_{n}^{2},\log (1/\delta)\}\log n }\right\}
    \label{eq: diff Rhat bound}
\end{align}
with probability tending to one.

Consider next $\sup_{g,g'\in\mathcal{G}:\|g-g'\|_{1}\le \varepsilon}
|\hat{R}_{2,n}^{\,g}-\hat{R}_{2,n}^{\,g'}|$.
Recall that, for each $i=1,\ldots,n$ and $k=1,\ldots,\Kf_i$ and $j \in S_{ik}$, we constructed a new variable by adding scaled i.i.d. standard normal noise $\zeta_{ikj}$,
$\hat W_{ikj}:=v_{ik}Z_i(\hat{\bm{\beta}},\hat \sigma) + d_{ikj}^{1/2} \hat \sigma 
\zeta_{ijk}~$.
Due to Lemma~\ref{lem:lemma1},
without any loss of generality we assume $\Kf_i=1$ for all $i=1,\ldots,n$. 
We denote $S_{ik}$ as $S_i$ as  $\Kf_i=1$.
Then we have
\begin{align*}
\hat R_{2,n}^{\,g} :=n^{-1}\sum_{i \in \mathcal{I}(h)}\sum_{j=1}^n s_i^{-1} \ex_{\zeta}\big\{\log \tilde m_g(\hat W_{ikj}(\zeta);u_i,v_{ik},\hat \sigma)\big\}\cdot\mathbf{1}\{j \in S_i\}.
\end{align*}
For $\delta>0$, define
\[\hat R_{2,n}^{\,g,\delta}=n^{-1}
\sum_{i \in \mathcal{I}(h)}\sum_{j=1}^n 
s_i^{-1} \ex_{\zeta}\big\{\log \max\{ \tilde m_g(\hat W_{ikj}(\zeta);u_{i},v_{ik},\hat{\sigma}),\delta\}\big\}
\cdot\mathbf{1}\{j \in S_i\},
\]
\begin{align*}
\begin{split}
\widehat{\mathcal{T}}_{2,\delta}
&:=
n^{-1}
\sup_{g\in\mathcal G}
\sum_{i \in \mathcal{I}(h)}\sum_{j=1}^n
s_i^{-1} \ex_{\zeta}
\Big{[}\left|\log \tilde m_g(\hat W_{ikj}(\zeta);u_{i},v_{ik},\hat{\sigma})\right|\\
&\qquad\qquad\qquad\qquad\qquad\qquad
\cdot\mathbf{1}\{ \tilde m_g(\hat W_{ikj}(\zeta);u_{i},v_{ik},\hat{\sigma})\le \delta,\,j\in S_{i}\}\Big{]},
\end{split}
\end{align*}
respectively.
Then, we have
\begin{align*}
    |\hat R_{2,n}^{\,g,\delta}-\hat R_{2,n}^{\,g',\delta}|
    \lesssim \frac{1}{\delta}\;\|g-g'\|_{1},
    \quad \text{and} \quad
    |\hat R_{2,n}^{\,g,\delta}-\hat R_{2,n}^{\,g}|
    \le \widehat{\mathcal{T}}_{2,\delta}.
\end{align*}
Recalling $s_{i}=|S_{i}|$,
we have
\[
\widehat{\mathcal{T}}_{2,\delta}
\le \sup_{g\in\mathcal{G}}
\sup_{i,k,j,u_{i},v_{ik}}
\ex_{\zeta}
\Big{[}\left|\log \tilde m_g(\hat W_{ikj}(\zeta);u_{i},v_{ik},\hat{\sigma})\right|
\cdot\mathbf{1}\{ \tilde m_g(\hat W_{ikj}(\zeta);u_{i},v_{ik},\hat{\sigma})\le \delta\}\Big{]}.
\]
By the same procedure as in bounding $\widehat{\mathcal{T}}_{1,\delta}$, 
the right hand side above is shown to have of the following order:
\begin{align*}
\widehat{\mathcal{T}}_{2,\delta}
\lesssim \delta^{b_{7}} \max\{a_{n},\sqrt{\log (1/\delta)}\} 
\sqrt{\log n}
\quad\text{for}\;\text{some}\; b_{7}>0
\end{align*}
with probability tending to one.
Therefore we obtain
\begin{align}
    \sup_{g,g'\in\mathcal{G}:\|g-g'\|_{1}\le \varepsilon}
\left|\hat{R}_{2,n}^{\,g}-\hat{R}_{2,n}^{\,g'}\right|
\lesssim \left\{\frac{\varepsilon}{\delta} + \delta^{b_{7}} \sqrt{\max\{a_{n}^{2},\log (1/\delta)\} \log n }\right\}
    \label{eq: diff Rhat bound 2}
\end{align}
with probability tending to one.

Combining \eqref{eq: diff R bound},
\eqref{eq: diff Rhat bound},
and
\eqref{eq: diff Rhat bound 2},
and taking 
\[\delta= \varepsilon^{\min\{1/(b_{4}^{2}/2),1/(b_{5}^{2}/2),1/(b_{6}+1),1/(b_{7}+1)\}}=:\varepsilon^{1/(b_{8}+1)},\]
we obtain
\begin{align}
\label{eq: diff Delta}
\sup_{g,g'\in\mathcal{G}:\|g-g'\|_{1}\le \varepsilon}
\left|\widehat{\Delta}_{n}[g]-\widehat{\Delta}_{n}[g']\right|
\lesssim \varepsilon^{b_{8}/(b_{8+1})}
\sqrt{\max\{a_{n}^{2},\log(1/\varepsilon)\} \log n}
.
\end{align}

\noindent \textbf{Step 2: 
Bounding the uniform deviation.}
To obtain a high-probability bound for 
$\max_{g\in\{g_{1},\ldots,g_{N(\varepsilon)}\}}
|\widehat{\Delta}_{n}[g]|$,
we need a sharper tail bound than the Chebyshev bound used in the proof of Theorem \ref{thm.risk.conv}.

\noindent \textbf{Step 2 (i):
Bounding the deviation $\hat R_{1,n}^{\,g}-R_{1,n}^{g}$}.
First, consider $\hat R_{1,n}^{\,g} - \ex[ \hat{R}^{g}_{1,n} ]$.
As in the proof of Theorem \ref{thm.risk.conv},
we set $\psi_{ik}:=\log m_g\!\big(v_{ik}Z_i(\hat{\bbeta},\hat{\sigma});u_i,v_{ik},\hat{\sigma}\big)$
and define $F_i:=\sum_{k=1}^{K_i}\psi_{ik}$: then we have
\[
\hat R_{1,n}^{\,g} - \ex[ \hat{R}^{g}_{1,n} ]
=\kappa_n^{-1}\sum_{i=1}^n \left\{F_i - \ex [F_{i}]\right\}.
\]
Here, from \eqref{eq: marginal bound 2},
the summand $F_i - \ex[F_{i}]$ has an independent sub-exponential envelope $K_{i} Z_{i}^{2}(\bbeta,\sigma)$ regardless of $g$, where $Z_{i}^{2}(\bbeta,\sigma)$ is non-centered $\chi^{2}$.
Note that the summands are uniformly sub-exponential, so is their normalized sum. 
Therefore a standard exponential deviation bound and the union bound
give
\begin{align*}
    \max_{g\in\{g_{1},\ldots,g_{N(\varepsilon)}\}}|\hat R_{1,n}^{\,g} - \ex[ \hat{R}^{g}_{1,n} ]| \lesssim n^{-1/2}\log n 
\end{align*}
with probability at least $1-\Pr(H_{n}^{\mathrm{c}})-N(\varepsilon)n^{-\nu}$ for any $\nu>0$.

Next, consider 
$\ex[ \hat{R}^{g}_{1,n} ]-R^{g}_{1,n}$.
Using the mean value theorem with \eqref{eq: log marginal derivative}, we have
\begin{align*}
\big|\mathbb{E}[\hat R_{1,n}^{\,g}]-R_{1,n}^{\,g}
\big|
&\lesssim
\kappa_n^{-1}
\sum_{i=1}^{n}\sum_{k=1}^{\tilde{K}_{i}}
\ex\big[
Z_{i}^{2}(\bbeta^{*},\sigma^{*})
\big(\|\hat{\bbeta}-\bbeta\|+|\hat\sigma-\sigma|\big)\big],
\end{align*}
where $(\bbeta^{*},\sigma^{*})$ lies on the line segment joining $(\hat{\bbeta},\hat{\sigma})$ and $(\bbeta,\sigma)$.
Here, by definition of $Z_{i}(\bbeta,\sigma)$,
we have
\[
Z_{i}(\bbeta^{*},\sigma^{*})=\frac{\sigma}{\sigma^{*}}Z_{i}(\bbeta,\sigma)
-\frac{1}{u_{i}}\sum_{k=1}^{K_{i}}u_{ik}
\frac{x_{ik}'(\bbeta^{*}-\bbeta)}{\sigma^{*}}.
\]
The triangle inequality, together with Assumptions~\hyperref[assumption1]{A1}
and~\hyperref[assumption2]{A2}, 
gives
\[
\left|Z_{i}(\bbeta^{*},\sigma^{*})
\right|
\lesssim Z_{i}(\bbeta,\sigma)
+\left\|\bbeta^{*}-\bbeta\right\|.
\]
This, together with the H\"{o}lder inequality ($p=1/4$ and $q=3/4$) and the conditional Gaussianity of $Z_{i}(\bbeta,\sigma)$,
gives
\begin{align*}
\big|\mathbb{E}[\hat R_{1,n}^{\,g}]-R_{1,n}^{\,g}
\big|
&\lesssim
\kappa_n^{-1}
\sum_{i=1}^{n}\sum_{k=1}^{\tilde{K}_{i}}
\ex\big[
Z_{i}^{2}(\bbeta,\sigma)
\big(\|\hat{\bbeta}-\bbeta\|+|\hat\sigma-\sigma|\big)\big]\\
&\lesssim \sqrt{\ex[\|\hat{\bbeta}-\bbeta\|^{2}]}
+\sqrt{\ex[\|\hat{\bbeta}-\bbeta\|^{4}]}
= O(n^{-\alpha}).
\end{align*}
Thus we obtain
\begin{align}
\label{eq: max R1}
    \max_{g\in\{g_{1},\ldots,g_{N(\varepsilon)}\}}|\hat R_{1,n}^{\,g} - R^{g}_{1,n} ]| \lesssim n^{-1/2}\log n +n^{-\alpha}
\end{align}
with probability at least $1-\Pr(H_{n}^{\mathrm{c}})-N(\varepsilon)n^{-\nu}$ for any $\nu>0$.

\noindent \textbf{Step 2 (ii):
Bounding the deviation $\hat R_{2,n}^{\,g}-R_{2,n}^{g}$}.
First, consider $\hat R_{2,n}^{\,g}-\ex[\hat R_{2,n}^{\,g}]$.
Observe that $\ex_{\zeta}\big\{\log \tilde m_g(\hat W_{ikj}(\zeta);u_i,v_{ik},\hat \sigma)\big\}$ has a sub-exponential envelope $E_{i}:=c Z_{i}^{2}(\bbeta,\sigma)+c'$ for some $c,c'>0$
that is independent across $i$. So, we have
\begin{align*}
\hat R_{2,n}^{\,g} - \ex[\hat{R}_{2,n}^{\,g}]\lesssim_{P} n^{-1}
\sum_{i \in \mathcal{I}(h)}
E_{i}
\underbrace{\Big{(}\sum_{j=1}^n \mathbf{1}\{j \in S_i\}/s_i\Big{)}}_{=:e_{i}},
\end{align*}
where $A \lesssim_{P} B$ denotes $\Pr(A \lesssim B)=1$.
Here the following lemma connects the coefficient $e_{i}$ to $D_{n}(h_{n})$:
\[
D_{n}(h_{n}):=\kappa_n^{-1}\sum_{i=1}^{n}\sum_{k=1}^{\Kf_{i}}\mathbf{1}\{|S_{ik}|\ge h_n\}/|S_{ik}|
=\kappa_{n}^{-1}\sum_{i\in \mathcal{I}(h)}1/s_{i}.
\]
\begin{lemma}\label{lem:combinatorialbound}
We have
\[
\sum_{j=1}^n e_j^2 \;\le\; \kappa_n^2\,D_n(h_n).
\]
\end{lemma}

\begin{proof}
Write $w_i:=\mathbf{1}\{s_{i}\ge h_n\}/s_{i}$. Then we have
\[
e_i=\sum_{j:\, j\in S_i} w_j
\quad\text{and}\quad
\sum_{i=1}^n e_i^2
=\sum_{i=1}^{n}\sum_{i'=1}^{n} w_i w_{i'}\,|S_i \cap S_{i'}|.
\]
By the Cauchy--Schwarz inequality, 
we have $|S_i \cap S_{i'}|\le \sqrt{|S_i|\,|S_{i'}|}$; hence we obtain
\[
\sum_{i=1}^n c_i^2
\le
\Big(\sum_{i=1}^{n} w_i \sqrt{|S_i|}\Big)^2
\le
\Big(\sum_{i=1}^{n} 1\Big)\Big(\sum_{i=1}^{n} w_i^2 |S_i|\Big)
=\kappa_n \sum_{i=1}^{n} \frac{\mathbf{1}\{s_{i}\ge h_n\}}{s_{i}}
=\kappa_n^2\,D_n(h_n).
\]
\end{proof}

Thus a standard exponential deviation bound and the union bound
give
\begin{align}
\label{eq: max R2}
    \max_{g\in\{g_{1},\ldots,g_{N(\varepsilon)}\}}|\hat R_{2,n}^{\,g} - \ex[ \hat{R}^{g}_{2,n} ]| \lesssim \sqrt{\frac{D_{n}(h_{n})}{n}}\log n 
\end{align}
with probability at least $1-\Pr(H_{n}^{\mathrm{c}})-N(\varepsilon)n^{-\nu}$ for any $\nu>0$.

Finally, noting that $R_{2,n}$ is the normalized sum of i.i.d.~bounded summands as in the proof of Theorem~\ref{thm.risk.conv} and considering the plug-in effect as in Step 2 (i),
we obtain 
\begin{align*}
    \max_{g\in\{g_{1},\ldots,g_{N(\varepsilon)}\}}|\hat R_{2,n}^{\,g} - R^{g}_{2,n} ]| \lesssim \sqrt{\frac{D_{n}(h_{n})}{n}}\log n +n^{-\alpha}
\end{align*}
with probability at least $1-\Pr(H_{n}^{\mathrm{c}})-N(\varepsilon)n^{-\nu}$ for any $\nu>0$.

\noindent \textbf{Final step}: 
Finally, we combine \eqref{eq: diff Rhat bound}, \eqref{eq: max R1}, and \eqref{eq: max R2}
to obtain
\[
\sup_{g\in\mathcal{G}}
\left|\widehat{\Delta}_{n}[g]\right|
\lesssim \left( c_{n}\log n + \varepsilon^{b_{8}/(b_{8}+1)} \sqrt{\max\{a_{n}^{2},\log(1/\varepsilon)\} \log n}\right)
\]
with probability tending to one.
Also, in the same procedure as in Step 2 (i),
we can bound $\sup_{g\in\mathcal{G}} |\widehat{\mathcal{P}}_{n}[g]|$:
\[
\sup_{g\in\mathcal{G}} |\widehat{\mathcal{P}}_{n}[g]|
\lesssim n^{-\alpha}.
\]
Now choose \(\varepsilon=\varepsilon_n=n^{-M}\) for a sufficiently large constant \(M>0\). By Assumption A8(i),
\[
\log N(\varepsilon_n,\mathcal G,\|\cdot\|_1)
\lesssim \log n,
\]
so taking \(M\) large enough yields
$
\varepsilon_n^{b_8/(b_8+1)}
=
n^{-Mb_8/(b_8+1)}
=
O(c_n )$
and the exclusion probability tends to zero.
Hence we obtain
\[
\sup_{g\in\mathcal G}|\widehat{\Delta}_n[g]|
=
O_p(c_n \, a_{n}\log n),
\]
which concludes the proof.

\subsubsection{Proof of Proposition~\ref{prop.A8}}

We verify Assumption~\hyperref[assumption8]{A8} for each class in turn.
The proof for \(\mathcal G_3\) is the same as that for \(\mathcal G_2\), and is therefore omitted.
Denote
\[
m_{g}(x;\sigma):=\int \phi(x \,;\,\gamma,\sigma ) g(\gamma)d\gamma.
\]
To verify Assumption~\hyperref[assumption8]{A8} (iii) for $\mathcal{G}_{1}$, the following identities are useful:
\begin{align}
\frac{\partial}{\partial x}\log m_{g}(x;\sigma)
&=-\frac{1}{\sigma}x
+ \frac{1}{\sigma}\frac{\int \gamma \phi(x\,;\,\gamma,\sigma) g(\gamma)d\gamma }{m_{g}(x;\sigma)},
\label{eq:marginalderivative_a}\\
\frac{\partial}{\partial \sigma}\log m_g(x;\sigma)
&=
-\frac{1}{2\sigma}
+
\frac{1}{2\sigma^2}
\frac{\int (x-\gamma)^2\phi(x;\gamma,\sigma)g(\gamma)\,d\gamma}{m_g(x;\sigma)}.
\label{eq:marginalderivative_sigma}
\end{align}

\medskip
\noindent
\textbf{Class \(\mathcal G_1\).}
Let
\[
g(\tau)=(1-\eta)\delta_0+\frac{\eta a}{2}e^{-a|\tau|},
\qquad
(\eta,a)\in [\underline{\eta},\overline{\eta}]\times[\underline a,\overline a],
\]
where \(0<\underline{\eta}\le \overline{\eta}<1\) and \(0<\underline a\le \overline a<\infty\).

We first verify Assumption~\hyperref[assumption8]{A8}(i). Since \(\mathcal G_1\) is parameterized by the two-dimensional compact rectangle
\([\underline{\eta},\overline{\eta}]\times[\underline a,\overline a]\), and the map
\((\eta,a)\mapsto g_{\eta,a}\) is continuous in \(L_1\), the class is parametric of two dimension and thus
$\log N(\varepsilon,\mathcal G_1,\|\cdot\|_1)\lesssim \log(\varepsilon^{-1})$, which proves ~\hyperref[assumption8]{A8}(i).

We next verify Assumption~\hyperref[assumption8]{A8}(ii).
Observe that, for any $M>0$, we have
\[
g([-M,M])
=
(1-\eta)+\eta\int_{-M}^M \frac{a}{2}e^{-a|\tau|}\,d\tau
=
1-\eta e^{-aM}.
\]
By choosing $M$ sufficiently large,
this yields 
\[
\inf_{g\in\mathcal G_1} g([-M,M])
\ge
1-\overline{\eta}e^{-\underline a M}=:\pi_{*} >0.
\]
Thus Assumption~\hyperref[assumption8]{A8}(ii) holds.

Finally, we verify Assumption~\hyperref[assumption8]{A8}(iii).
Denote
\[
g(\gamma)=(1-\eta)\delta_0+\eta h_a(\gamma)
\qquad\text{with}\quad
h_a(\gamma):=\frac{a}{2}e^{-a|\gamma|},
\]
where $(\eta,a)\in[\underline\eta,\overline\eta]\times[\underline a,\overline a]$,
with $0<\underline\eta\le \overline\eta<1$ and 
$0<\underline a\le \overline a<\infty$.
Then we have
\[
m_g(x;\sigma)
=
(1-\eta)\phi(x;0,\sigma)+\eta\int \phi(x;\gamma,\sigma)h_a(\gamma)\,d\gamma,
\]
from which we have a lower bound on $m_{g}$:
\[
m_g(x;\sigma)\ge (1-\eta)\phi(x;0,\sigma)\ge (1-\overline\eta)\phi(x;0,\sigma).
\]

From \eqref{eq:marginalderivative_a} and \eqref{eq:marginalderivative_sigma},
it suffices to control the posterior moments up to second.
For $r=0,1,2$, let
\[
I_r(x;\sigma,a):=\int |\gamma|^r \phi(x;\gamma,\sigma)h_a(\gamma)\,d\gamma.
\]
From the lower bound
\[
m_g(x;\sigma)=(1-\eta)\phi(x;0,\sigma)+\eta I_0(x;\sigma,a)\ge \eta I_0(x;\sigma,a),
\]
we obtain
\[
\frac{\int |\gamma|^r \phi(x;\gamma,\sigma)g(\gamma)\,d\gamma}{m_g(x;\sigma)}
=
\frac{\eta I_r(x;\sigma,a)}{(1-\eta)\phi(x;0,\sigma)+\eta I_0(x;\sigma,a)}
\le
\frac{I_r(x;\sigma,a)}{I_0(x;\sigma,a)}.
\]
It therefore suffices to bound $I_r/I_0$ uniformly over $a\in[\underline a,\overline a]$.

To bound $I_r/I_0$, decompose
$I_r(x;\sigma,a)=I_r^+(x;\sigma,a)+I_r^-(x;\sigma,a)$,
where
\[
I_r^+(x;\sigma,a):=\int_0^\infty \gamma^r \phi(x;\gamma,\sigma)h_a(\gamma)\,d\gamma,
\quad\text{and}\quad
I_r^-(x;\sigma,a):=\int_{-\infty}^0 |\gamma|^r \phi(x;\gamma,\sigma)h_a(\gamma)\,d\gamma.
\]
Since $I_0=I_0^++I_0^-$, we have
\[
\frac{I_r}{I_0}
=\frac{I_r^+ + I_{r}^{-}}{I_0^+ + I_{0}^{-}}
\le
\max\left\{\frac{I_r^+}{I_0^+},\frac{I_r^-}{I_0^-}\right\}.
\]
Thus it suffices to bound $I_r^\pm/I_0^\pm$.

For the positive half-line integral, using \(h_a(\gamma)=(a/2)e^{-a\gamma}\) for \(\gamma\ge0\),
\[
I_r^+(x;\sigma,a)
=
\frac{a}{2}\int_0^\infty \gamma^r \phi(x;\gamma,\sigma)e^{-a\gamma}\,d\gamma.
\]
Completing the square,
\[
-\frac{(x-\gamma)^2}{2\sigma}-a\gamma
=
-\frac{1}{2\sigma}\bigl(\gamma-(x-\sigma a)\bigr)^2
+
\underbrace{\frac{\sigma a^2}{2}-ax}_{\text{independent of $\gamma$}},
\]
we notice that the common factor \(e^{-ax+\sigma a^2/2}\) cancels in the ratio. Hence, with
\[
\mu_+(x,a)=\mu_{+}:=x-\sigma a,
\]
we obtain
\[
\frac{I_r^+(x;\sigma,a)}{I_0^+(x;\sigma,a)}
=
\frac{\int_0^\infty \gamma^r \exp\!\left[-(\gamma-\mu_+)^2/(2\sigma)\right]\,d\gamma}
{\int_0^\infty \exp\!\left[-(\gamma-\mu_+)^2/(2\sigma)\right]\,d\gamma}.
\]
After the change of variables \(\gamma=\mu_+ + \sqrt{\sigma}\,z\), this becomes
\[
\frac{I_r^+}{I_0^+}
=
\frac{\int_{-\mu_+/\sqrt{\sigma}}^\infty |\mu_+ + \sqrt{\sigma}\,z|^r e^{-z^2/2}\,dz}
{\int_{-\mu_+/\sqrt{\sigma}}^\infty e^{-z^2/2}\,dz}.
\]
Using $|u+v|^r\le 2(|u|^r+|v|^r)$ ($r=1,2$), it follows that
\[
\frac{I_r^+}{I_0^+}
\le
2 |\mu_+|^r
+
2 \sigma^{r/2}
\frac{\int_{-\mu_+/\sqrt{\sigma}}^\infty |z|^r e^{-z^2/2}\,dz}
{\int_{-\mu_+/\sqrt{\sigma}}^\infty e^{-z^2/2}\,dz}.
\]
Here, if
$t:=-\mu_+/\sqrt{\sigma}\le 0$, then the denominator is bounded below by
\(\int_0^\infty e^{-z^2/2}\,dz>0\), whereas the numerator is finite from the Gaussian integral. If $t>0$, then by the standard Mills ratio bounds,
\[
\int_t^\infty z e^{-z^2/2}\,dz = e^{-t^2/2}
\lesssim
t \int_t^\infty e^{-z^2/2}\,dz,
\]
and
\[
\int_t^\infty z^2 e^{-z^2/2}\,dz
=
t e^{-t^2/2}+\int_t^\infty e^{-z^2/2}\,dz
\lesssim
(1+t^2)\int_t^\infty e^{-z^2/2}\,dz.
\]
Therefore,
\[
\frac{I_1^+}{I_0^+} \lesssim (1+|\mu_+|),
\qquad
\frac{I_2^+}{I_0^+} \lesssim (1+\mu_+^2).
\]

The negative half-line integral is handled similarly and we obtain
\[
\frac{I_1^-(x;\sigma,a)}{I_0^-(x;\sigma,a)}\lesssim (1+|\mu_-|),
\qquad
\frac{I_2^-(x;\sigma,a)}{I_0^-(x;\sigma,a)}\lesssim (1+\mu_-^2),
\]
where $\mu_{-}:= -x-\sigma a$

Since $a\in[\underline a,\overline a]$, we have
\[
|\mu_+(x,a)|+|\mu_-(x,a)| \lesssim (1+|x|),
\qquad
\mu_+(x,a)^2+\mu_-(x,a)^2 \lesssim (1+x^2).
\]
Consequently, we have
\[
\frac{I_1(x;\sigma,a)}{I_0(x;\sigma,a)}\lesssim (1+|x|),
\qquad
\frac{I_2(x;\sigma,a)}{I_0(x;\sigma,a)}\lesssim (1+x^2).
\]

Thus, we obtain
\[
\left|
\frac{\partial}{\partial x}\log m_g(x;\sigma)
\right|
\le
\frac{|x|}{\sigma}
+
\frac{1}{\sigma}\frac{I_1(x;\sigma,a)}{I_0(x;\sigma,a)}
\lesssim (1+|x|),
\]
and, using \((x-\gamma)^2\le 2x^2+2\gamma^2\),
\[
\left|
\frac{\partial}{\partial \sigma}\log m_g(x;\sigma)
\right|
\le
\frac{1}{2\sigma}
+
\frac{x^2}{\sigma^2}
+
\frac{1}{\sigma^2}\frac{I_2(x;\sigma,a)}{I_0(x;\sigma,a)}
\lesssim (1+x^2).
\]
This proves Assumption~A8(iii) for \(\mathcal G_1\), which completes the proof for $\mathcal{G}_{1}$.

\medskip
\noindent
\textbf{Class \(\mathcal G_2\).}
Let
\[
g(\tau)=\sum_{l=1}^L \pi_l \phi(\tau;0,\nu_l),
\qquad
\pi_l\ge 0,\quad \sum_{l=1}^L \pi_l=1,
\]
where \(\nu_1,\dots,\nu_L\) is a fixed variance grid.

Since the weight vector $(\pi_1,\dots,\pi_L)$ ranges over the simplex, $\mathcal G_2$ is a compact parametric family of dimension \(L-1\). Therefore
$\log N(\varepsilon,\mathcal G_2,\|\cdot\|_1)\lesssim \log(\varepsilon^{-1})$, and A8(i) follows.

Also,
\[
g([-M,M])
=
\sum_{l=1}^L \pi_l \int_{-M}^M \phi(\tau;0,\nu_l)\,d\tau
\ge
\min_{1\le l\le L}\int_{-M}^M \phi(\tau;0,\nu_l)\,d\tau,
\]
which proves A8(ii).

For A8(iii), write
\[
m_g(x;\sigma)
=
\int \phi(x;\tau,\sigma)g(\tau)\,d\tau
=
\sum_{l=1}^L \pi_l \phi(x;0,\sigma+\nu_l).
\]
The derivative of $\log m_{g}$ with respect to $x$,
\[
\frac{\partial}{\partial x} \log m_g(x;\sigma)
=
-x\,
\frac{\sum_{l=1}^L \pi_l (\sigma+\nu_l)^{-1}\phi(x;0,\sigma+\nu_l)}
{\sum_{l=1}^L \pi_l \phi(x;0,\sigma+\nu_l)},
\]
gives
\[
|(\partial/\partial x) \log m_g(x;\sigma)|
\le
|x|\max_{1\le l\le L}(\sigma+\nu_l)^{-1}.
\]
Likewise, we have
\[
\left|\frac{\partial}{\partial \sigma} \log m_g(x;\sigma)\right|
\le
\max_{1\le l\le L}
\left\{
\frac{1}{2(\sigma+\nu_l)}
+
\frac{x^2}{2(\sigma+\nu_l)^2}
\right\}.
\]
Therefore we obtain
\[
|\partial_x \log m_g(x;\sigma)|+|\partial_\sigma \log m_g(x;\sigma)|
\lesssim (1+|x|+|x|^2),
\]
uniformly over \(g\in\mathcal G_2\). Thus A8(iii) holds.

Combining these completes the proof.

\bibliographystyle{abbrvnat}
\bibliography{mybib,paperref,pred-inf,append}

\end{document}